\documentclass[11pt,a4paper]{article}
\usepackage{a4wide}
\usepackage{epsf}\epsfverbosetrue
\usepackage{indentfirst}
\usepackage{subfigure}
\usepackage{amstext}
\usepackage{rotating}
\usepackage{amsmath}
\usepackage{amssymb}
\usepackage{amsfonts}
\usepackage{lscape}
\usepackage{xspace}
\usepackage{color,hyperref}

\usepackage{graphicx}
\usepackage{marginnote}
\usepackage{multirow}
\usepackage{algorithmic}
\usepackage{algorithm}
\usepackage{tikz}
\usepackage{pgf}
\usetikzlibrary[positioning]
\usepackage[thmmarks]{ntheorem}
\definecolor{red}{rgb}{1.0,0.0,0.0}
\definecolor{blue}{rgb}{0.0,0.0,1.0}
\definecolor{green}{rgb}{0.0,1.0,0.0}
\hypersetup{colorlinks,breaklinks,
           linkcolor=red,urlcolor=blue,
           anchorcolor=red,citecolor=green}

\theoremstyle{plain}
\newtheorem{definition}{Definition}

\theoremstyle{plain}
\newtheorem{theorem}{Theorem}

\theoremstyle{plain}
\newtheorem{example}{Example}

\theoremstyle{plain}
\newtheorem{lemma}{Lemma}

\theorembodyfont{\upshape}
\theoremstyle{nonumberplain} \theoremseparator{}
\theoremsymbol{\rule{1ex}{1ex}}
\newtheorem{proof}{Proof}

\def\p2p{peer-to-peer}
\def\P2p{Peer-to-peer}

\newcommand{\sig}{\ensuremath{\mathsf{Sig}}\xspace}

\newcommand{\rndadv}{\ensuremath{\mathsf{RG\textendash Adv}}\xspace}  
  
\newcommand{\statadv}{\ensuremath{\mathsf{ST\textendash Adv}}\xspace}  
  
\newcommand{\bstatadv}{\ensuremath{\mathsf{Bayesian\textendash ST\textendash Adv}}\xspace}  
\newcommand{\badv}{\ensuremath{\mathsf{Bayesian\textendash Adv}}\xspace}  

\begin{document}

\title{DREAM: DiffeRentially privatE smArt Metering \\ 
{\large Technical report}}

\author{Gergely Acs and Claude Castelluccia \\
INRIA Rhone Alpes, Montbonnot, France \\
\texttt{{gergely.acs, claude.castelluccia}@inrialpes.fr} \\ }
\date{}

\maketitle
\sloppy

 \setlength{\textfloatsep}{0.3cm}
\begin{abstract}
This paper presents a new privacy-preserving smart metering system.
Our scheme is private under the differential privacy model
and therefore provides strong and provable  guarantees. 
With our scheme, an (electricity) supplier can periodically collect data from smart
meters and derive aggregated statistics while learning only limited information about
the activities of individual households. 
For example, a supplier cannot 
tell from a user's trace when he watched TV or turned on heating. 
Our scheme is simple, efficient and practical. Processing cost is very limited:  
smart meters only have to add noise to their data and encrypt the results with 
an efficient stream cipher.
\end{abstract}

\newpage
\tableofcontents
\newpage

\section{Introduction}

Several countries throughout the world are planning to deploy smart meters in households in the very near future.
The main motivation, for governments and electricity suppliers, is to be able to match consumption with generation. 
Traditional electrical meters only measure total consumption on a given period of time (i.e., one month or one year). As such, 
they do not provide accurate information of when the energy was consumed. 
Smart meters, instead, monitor and report consumption in intervals of few minutes. 
They allow the utility provider to monitor, almost in real-time, consumption and possibly adjust generation and prices according 
to the demand. Billing customers by how much is consumed and at what time of day will probably change consumption habits to 
help matching consumption with generation. In the longer term, with the advent of smart appliances, it is expected that the smart grid 
will remotely control selected appliances to reduce demand.

\paragraph{Problem Statement:}
Although smart metering might help improving energy management, it creates many new privacy problems \cite{anderson10weis}.
Smart meters provide very accurate consumption data to electricity providers.
As the interval of data collected by smart meters decreases, 
the ability to disaggregate low-resolution data increases.  
Analyzing high-resolution consumption data, Nonintrusive Appliance Load Monitoring (NALM) \cite{hart92ieee} can be used to identify a remarkable number of electric
appliances (e.g., water heaters, well pumps, furnace
blowers, refrigerators, and air conditioners) employing exhaustive appliance signature libraries. Researchers are now focusing on the myriad
of small electric devices around the home such as personal computers,
laser printers, and light bulbs \cite{Lam07trans}. Moreover, it has also been shown that even simple off-the-shelf statistical tools can be used to extract complex usage patterns from high-resolution consumption data \cite{Molina10buildsys}.
This extracted information can be used to profile and monitor users
for various purposes, creating serious privacy risks and concerns.  
As data recorded by smart
meters is lowering in resolution, and inductive algorithms are quickly improving, it is urgent
to develop privacy-preserving smart metering systems that provide strong and provable 
guarantees.

%
%

\paragraph{Contributions:}
We propose a privacy-preserving smart metering scheme that guarantees users' privacy 
while still preserving the benefits and promises of smart metering.
Our contributions are many-fold and summarized as follows:
\begin{itemize}
\item We provide the first provably private and distributed solution for smart metering that
optimizes utility without relying on a third trusted party (i.e., an 
aggregator). 
We were able to avoid the use of a third trusted party by proposing a new distributed 
Laplacian Perturbation Algorithm (DLPA). 

In our scheme, smart meters are grouped into clusters, where a cluster is a group of hundreds 
or thousands of smart meters corresponding, for example,  to a quarter of a city. Each smart meter sends, at each sampling
period, their measures to the supplier. 
These measures are noised and encrypted such that the supplier can compute
the noised aggregated electricity consumption of the cluster, at each sampling period, without getting access to 
individual values.  The aggregate is noised just enough to provide differential privacy to each participating user,  
while still providing high utility (i.e., low error).
Our scheme is secure under the differential privacy model
and therefore provides strong and provable privacy guarantees. 
In particular, we guarantee that the supplier can retrieve information about any user consumption only up to a predefined threshold. 
Our scheme is simple, efficient and practical. It requires either one or two rounds of message exchanges
between a meter and the supplier. Furthermore, processing cost is very limited:  smart meters only have 
to add noise to their data and encrypt the results with an efficient stream cipher. Finally, our scheme is robust against
smart meter failures and malicious nodes. More specifically, it is secure even if an $\alpha$ fraction of all nodes 
of a cluster collude with the supplier, where $\alpha$ is a security parameter.

\item We provide a detailed analysis of the security and performance of our proposal.
The security analysis is performed analytically. The performance, which is evaluated using
the utility metric, is performed using simulation. We implemented
a new electricity trace generation tool based on \cite{richardson10} which generates one-minute resolution synthetic consumption data of different households. 

\end{itemize}

\section{Related Work}

Several papers addressed the privacy problems of smart metering in the recent past \cite{kalogridis10smartgridcommA,Molina10buildsys,anderson10weis,anderson10key,anderson10smartgridcomm,bohli10icc,danezis10tc,garcia10stm}. However, only a few of them have proposed technical solutions to protect users' privacy. 
In \cite{anderson10weis,anderson10smartgridcomm}, the authors discuss the different security aspects of smart metering and the conflicting interests among stakeholders. 
The privacy of billing is considered in \cite{danezis10tc,Molina10buildsys}. These techniques uses zero-knowledge proofs to ensure that the fee calculated by the user is correct without disclosing any consumption data.   

Seemingly, the privacy of monitoring the sum consumption of multiple users may be solved by simply anonymizing individual measurements like in \cite{kalogridis10smartgridcommA} or using some mixnet. However, these ``ad-hoc'' techniques are dangerous and do not provide any real assurances of privacy. Several prominent examples in the history have shown that ad-hoc methods do not work \cite{korolova09www}. Moreover, 
these techniques require an existing trusted third party who performs anonymization.    
The authors in \cite{bohli10icc} perturb the released aggregate with random noise and use a different model from ours to analyze the privacy of their scheme. However, they do not encrypt individual measurements which means that the added noise must be large enough to guarantee reasonable privacy. As individual noise shares sum up at the aggregation, the final noise makes the aggregate useless.  
In contrast to this, \cite{garcia10stm} uses homomorphic encryption to guarantee privacy for individual measurements. However, the aggregate is not perturbed which means that it is not differential private. 

The notion of differential privacy was first proposed in \cite{dwork06tcc}. 
The main advantage of differential privacy over other privacy models is that it does not specify the prior knowledge of the adversary and provides rigorous privacy guarantee if each users' data is statistically independent \cite{kifer11sigmod}.
Initial works on differential privacy focused on the problem how a trusted curator (aggregator), who
collects all data from users, can differential privately release statistics. By contrast,
our scheme ensures differential privacy even if the curator is untrusted. 
Although \cite{dwork06eurocrypt} describes protocols for generating shares of random noise which is secure against malicious participants, it requires communication between users and it uses expensive secret sharing techniques resulting in high overhead in case of large number of users. 
Similarly, traditional Secure Multiparty Computation (SMC)
techniques \cite{GoldreichMulitParty} \cite{Cramer01eurocrypt} also require interactions between users. All these solutions are impractical for resource constrained smart meters where all the computation is done by the aggregator and users are not supposed to communicate with each other.

Two closely related works to ours are \cite{rashtogi10sigmod} and \cite{shi11ndss}.
In \cite{rashtogi10sigmod}, the authors propose a scheme to differential privately aggregate sums over multiple slots when the aggregator is untrusted. However, they use the threshold Paillier cryptosystem \cite{fougue00paillier} for homomorphic encryption which is much more expensive compared to  \cite{cc05mobiquitous} that we use.
They also use different noise distribution technique which requires several rounds of message exchanges between the users and the aggregator. 
By contrast, our solution is much more efficient and simple: it requires only a single message exchange if there are no node failures, otherwise, we only need one extra round. In addition, our solution does not rely on expensive public key cryptography during aggregation.  

A recent paper \cite{shi11ndss} proposes another technique to privately aggregate time series data. This work differs from ours as follows: 
(1) they use a Diffie-Hellman-based encryption scheme, whereas our construction is based on a more efficient construction that only use modular
 additions. This approach is better adapted to resource constrained devices like smart meters.   
(2) Although \cite{shi11ndss} does not require the establishment (and storage) of pairwise keys between nodes as opposed to our approach, 
it is unclear how \cite{shi11ndss} can be extended to tolerate node and communication failures. By contrast, our scheme is more robust, as the encryption key of non-responding nodes is known to other nodes in the network that can help to recover the aggregate. 
(3) Finally, \cite{shi11ndss} uses a different noise generation method from ours, but this technique only satisfies the relaxed $(\varepsilon, \delta)$-differential privacy definition. Indeed, in their scheme, each node adds noise probabilistically which means that none of the nodes add noise with some positive probability $\delta$. 
Although $\delta$ can be arbitrarily small, this also decreases the utility. By contrast, in our scheme, $\delta=0$ while ensuring nearly optimal utility.

\section{The model}

\subsection{Network model}

The network is composed of four major parts: the \emph{supplier/aggregator}, the \emph{electricty distribution network}, the \emph{communication network}, and the \emph{users} (customers). 
Every user is equipped with an electricity smart meter, which measures the electricity consumption of the user in every $T_{p}$ long period, and, using the communication network,  
sends the measurement to the aggregator at the end of every slot (in practice, $T_{p}$ is around 1-30 minutes). Note that the communication and distribution network can be the same (e.g., when PLC technology is used to transfer data). 
The measurement of user $i$ in slot $t$ is denoted by $X_{t}^{i}$. The consumption profile of user $i$ is described by the vector $(X_{1}^{i}, X_{2}^{i},  \ldots)$, where the measurements of different users
are statistically independent.  
Privacy directly correlates with
$T_{p}$; finer-grained samples means more accurate profile, but also entails weaker privacy.
The supplier is interested in the sum of all measurements in every slot (i.e., $\sum_{i=1}^{N}X_{t}^i \stackrel{\mathsf{def}}{=} \mathbf{X}_{t}$).

As in \cite{bohli10icc}, we also assume that smart meters are trusted devices (i.e., tamper-resistant) which can store key materials and perform crypto computations. This realistic assumption has also been confirmed in \cite{anderson10smartgridcomm}. We assume that each node is configured with a private key and gets the corresponding certificate from a trusted third party. For example, each country might have a third party that generates these certificate and can additionally generate the ``supplier'' certificates to supplier companies \cite{anderson10smartgridcomm}.
As in \cite{anderson10smartgridcomm}, we also assume that public key operations are employed only for initial key establishment, probably when a meter is taken over by a new supplier. Messages exchanged between the supplier and the meters are authenticated using pairwise MACs \footnote{Please refer to \cite{anderson10key} for a more detailed
discussion about key management issues in smart metering systems.}.  Smart meters are assumed to have bidirectional communication channel (using some wireless or PLC technology) with the aggregator, but the meters cannot communicate with each other. 
We suppose that nodes may (randomly) fail, and in these cases, cannot send their measurements to the aggregator. However, nodes are supposed to use some reliable transport protocol to overcome the transient communication failures of the channel. Finally, we note that smart meters also allow the supplier to perform fine-grained billing based on time-dependant variable tariffs. Here, we are not concerned with the privacy and security problems of this service. Interested readers are referred to \cite{danezis10tc,Molina10buildsys}.

\subsection{Adversary model}
\label{sec:adv_model}

In general, the objective of the adversary is 
to infer detailed information about household
activity (e.g, how many people are in home and 
what they are doing at a given time).
In order to do that, it needs to extract complex usage patterns of appliances which 
include the level of power consumption, periodicity, and duration. 
It has been shown in \cite{Molina10buildsys} that different data mining techniques can be easily 
applied to a raw consumption profile to obtain this information.

%
%

In terms of its capability, we distinguish three types of adversary. The first is the 
a \emph{honest-but-curious (HC) adversary}, who attempts to obtain private 
information about a user, but it follows the protocol faithfully and do
not provide false information \cite{Molina10buildsys}. It only uses the (non-manipulated) 
collected data.

The \emph{dishonest-but-non-intrusive (DN) adversary} may not 
follow the protocol correctly and is allowed
to provide false information to manipulate the collected data. 
Some users can also be malicious and collude even 
with the supplier to collect information about honest users. 
However, the DN adversary is not allowed to access and modify the distribution network
to mount attacks. In particular, he is not allowed to install wiretapping devices to eavesdrop on the victim's consumption.

Likewise the DN adversary, the strongest \emph{dishonest-and-intrusive (DI) adversary} may not follow
all protocols either, but that can, in addition, invade the distribution network to gather more information about clients. In other words, the DI adversary can monitor the electricity consumption of the clients by installing meters on the power line that is outside of the client's control (like outside from his household). 

We suppose that all types of adversary can have any kind of extra knowledge about honest users, beyond the collected measurements, which might help to infer private information about them. For instance, it can observe their daily activities\footnote{Similarly to monitoring neighbors. Indeed, neighbors can also be malicious users, which is included in our model.}, or obtain extra information by doing personal interviews, surveys, etc. 

\subsection{Privacy model}
We use differential privacy \cite{dwork06tcc} that models the adversary described above. 
In particular, differential privacy guarantees that a user's privacy should not be threatened substantially more if he provides his measurement to the supplier. 

\begin{definition}[$\varepsilon$-differential privacy]
\label{def:diff_priv_stand}
An algorithm $\mathcal{A}$ is $\varepsilon$-differential private, if for all
data sets $D_{1}$ and $D_{2}$, where $D_{1}$ and $D_{2}$ differ in at most a single user,
and for all subsets of possible answers $S \subseteq \mathit{Range}(\mathcal{A})$,
$$
P(\mathcal{A}(D_{1}) \in S) \leq e^{\varepsilon} \cdot P(\mathcal{A}(D_{2}) \in S) 
$$
\end{definition}

Differential private algorithms produce indistinguishable outputs for similar inputs (more precisely, differing by a single entry), and thus, the modification of any single user's data in the dataset changes the probability of any output only up to a multiplicative factor $e^\varepsilon$. The parameter $\varepsilon$ allows us to control the level of privacy. Lower values of $\varepsilon$ implies stronger privacy, as they restrict further the influence of a user's data on the output. 
Note that this model, if users' data are independent, guarantees privacy for a user even if all other users' data is known to the adversary (e.g., it knows all measurements comprising the aggregate except the target user's), like when $N-1$ out of $N$ users are malicious and cooperate with the supplier. 

\begin{example}[Illustration of $\varepsilon$-differential privacy]
There is a dataset $D$ containing a list of patients' entries. Each entry has an attribute that indicates whether the corresponding patient has cancer or not. Suppose an $\varepsilon$-differential private query $\mathcal{A}$ that returns the sanitized number of patients in $D$ that have cancer. We assume that the adversary knows the exact number of cancer patients, $x$, before adding Alice to $D$, and wants to learn from the  random output $O$ of $\mathcal{A}(D\cup\{\text{Alice}\})$ whether Alice has cancer or not. The adversary has no prior knowledge about Alice (i.e., the probability that Alice has cancer is 0.5 before accessing $O$). The adversary either infers Alice as a cancer or a non-cancer patient.   
The success probability of this inference has a maximum of $\frac{1}{1+e^{-\varepsilon}}$ (and $\geq 0.5$)\footnote{Let $A$ denote the event that Alice has cancer. Using a bayesian reasoning, $P(A|O) = \frac{P(O|A)}{P(O|A)+P(O|\overline{A})} = \frac{P(\mathcal{A}(x+1)=O)}{P(\mathcal{A}(x+1)=O)+ P(\mathcal{A}(x)=O)} \leq \frac{1}{1+e^{-\varepsilon}}$, where we used that $P(A)=P(\overline{A})$ and  $e^{-\varepsilon} \leq \frac{P(\mathcal{A}(x)=O)}{P(\mathcal{A}(x+1)=O)} \leq e^\varepsilon$. Moreover, the optimal inference strategy is the maximum likelihood decision: the adversary infers Alice as a cancer patient if $P(\mathcal{A}(x+1) = O) > P(\mathcal{A}(x) = O)$ or with probability 0.5 if $P(\mathcal{A}(x+1) = O) = P(\mathcal{A}(x) = O)$, otherwise as a non-cancer patient.}. For example, the values 2, 1, 0.5, 0.1 of $\varepsilon$ yield correct inferences with a maximum probability of 0.88, 0.73, 0.62, 0.52, resp.   
\end{example}

The definition of differential privacy also maintains a \emph{composability property}: the composition of differential private algorithms remains differential private and their $\varepsilon$ parameters are accumulated. In particular, a protocol having $t$ rounds, where each round is individually $\varepsilon$ differential private, is itself $t\cdot\varepsilon$ differential private.


\subsection{Output perturbation: achieving differential privacy}
\label{sec:perturbation}

Let's say that we want to publish in a differentially private way the output of a function $f$.
The following theorem says that this goal can be achieved by perturbing the output of $f$; simply adding a random noise to the value of $f$, where the noise distribution is carefully calibrated to the global sensitivity of $f$, results in $\varepsilon$-differential privacy. The global sensitivity of a function is the maximum "change" in the value of the function when its input differs in a single entry. For instance, if $f$ is the sum of all its inputs, the sensitivity is the maximum value that an input can take.

\begin{theorem}[Laplacian Perturbation Algorithm (LPA) \cite{dwork06tcc}]
\label{thm:histo}
For all $f : \mathbb{D} \rightarrow \mathbb{R}^{r}$, the following mechanism $\mathcal{A}$ is $\varepsilon$-differential private: $\mathcal{A}(D)= f(D) + \mathcal{L}(S(f)/\varepsilon)$, where $\mathcal{L}(S(f)/\varepsilon)$ is an independently generated random variable following the Laplace distribution and $S(f)$ denotes the global sensitivity of $f$\footnote{Formally, let $f : \mathbb{D} \rightarrow \mathbb{R}^{r}$, then the global sensitivity of $f$ is $S(f) = \max ||f(D_{1}) - f(D_{2})||_{1}$, where $D_{1}$ and $D_{2}$ differ in a single entry and $|| \cdot ||_{1}$ denotes the $L_{1}$ distance.}.
\end{theorem}



\begin{example}
\label{ex:1}
To illustrate these definitions, consider a mini smart metering application, where users $U_1$, $U_2$, and $U_3$ need to send the sum of their measurements in two consecutive slots. The measurements of $U_1$, $U_2$ and $U_3$ are $(X_{1}^1 = 300, X_{2}^1 = 300)$, $(X_{1}^2 = 100, X_{2}^2 = 400)$, and $(X_{1}^3 = 50, X_{2}^3 = 150)$, resp. The nodes want differential privacy for the released sums with at least a $\varepsilon=0.5$. Based on Theorem \ref{thm:histo}, they need to add $\mathcal{L}(\lambda = \max_i \sum_{t} X_t^i / 0.5 = 1200)$ noise to the released sum in \textbf{each} slot. This noise ensures $\varepsilon = \sum_t X_{t}^1 / \lambda = 0.5$ individual indistinguishability for $U_1$, $\varepsilon = 0.42$ for $U_2$, and $\varepsilon = 0.17$ for $U_3$. Hence, the global $\varepsilon=0.5$ bound is guaranteed to all. Another interpretation is that $U_1$ has $\varepsilon_1 = X_{1}^1 / \lambda =0.25$, $\varepsilon_2=X_{2}^1 / \lambda=0.25$ privacy in each individual slot, and $\varepsilon = \varepsilon_1 + \varepsilon_2=0.5$ considering all two slots following from the composition property of differential privacy.   
\end{example}

\subsection{Utility definition}
Let $f : \mathbb{D} \rightarrow \mathbb{R}$.
In order to measure the utility, we quantify the difference between $f(D)$ and its perturbed value (i.e., $\hat{f}(D) = f(D) +  \mathcal{L}(\lambda)$) which is the error introduced by LPA. A common scale-dependant error measure is the Mean Absolute Error (MAE), which is 
$\mathbb{E}| f(D) - \hat{f}(D)|$ in our case.
However, the error should be dependent on the non-perturbed value of $f(D)$; if
$f(D)$ is greater, the added noise becomes
small compared to $f(D)$ which intuitively results in better utility. Hence, we rather use a slightly
modified version of a scale-independent metric called Mean Absolute Percentage Error (MAPE), which shows the proportion of the error to the data, as follows.

\begin{definition}[Error function]
\label{def:utility}
Let $D_{t} \in \mathbb{D}$ denote a dataset in time-slot $t$.
Furthermore, let $\delta_{t} = \frac {| f(D_{t}) - \hat{f}(D_{t}) |}{f(D_{t})+1}$ (i.e., the value of the error in slot $t$). The error function is defined 
as $\mu(t) = \mathbb{E}(\delta_{t})$. The expectation is taken on the randomness of $\hat{f}(D_{t})$. 
The standard deviation of the error is $\sigma(t) = \sqrt{\mathit{Var}(\delta_{t})}$ in time $t$. 
\end{definition}


In the rest of this paper, the terms "utility" and "error" are used interchangeably.


\section{Objectives}

Our goal is to develop a practical scheme that should not introduce more privacy risks
for users than traditional metering systems while retaining the benefits of smart meters. More specifically, the scheme should be
\begin{itemize}
\item \emph{differentially private}:  Considering \textbf{DN adversary}, the scheme differential privately releases sanitized aggregates $\hat{\mathbf{X}}_{t}$ where the leaked information about users is measured by $\varepsilon$.
\item \emph{robust and easily configurable}: It tolerates (random) node failures.
\item \emph{efficient}: It has low overhead which includes low computation load on smart meters, and low communication overhead between the supplier and individual meters. It should use pubic key operations only for initial key establishment. Afterwards, all communication is protected using more efficient symmetric crypto-based techniques. 
\item \emph{distributed}: Besides a certificate authority, the protocol does not require any trusted third party such as a trusted aggregator as in \cite{bohli10icc}. The smart meters communicate directly with the supplier.
\item \emph{useful for the supplier}: The
sanitized and the original (non-sanitized) aggregate should be ``similar'' (i.e., the error should be as small as possible). For instance, the supplier should be able to perform efficient management of the resource using the sanitized data: to monitor the consumption
at the granularity of a maximum few hundred households, and to detect consumption peaks or abnormal consumption. 
\end{itemize}

\section{Overview of approaches}

Our task is to enable the supplier to calculate the sum of maximum $N$ measurements (i.e., $\sum_{i=1}^N X_{t}^{i} = \mathbf{X}_{t}$ in all $t$) coming from $N$ different
users while ensuring $\varepsilon$-differential privacy for each user. 
This is guaranteed if the supplier can only access $\mathbf{X}_{t} + \mathcal{L}(\lambda(t))$, where $\mathcal{L}(\lambda(t))$ \footnote{We will use the notation $\lambda$ instead of $\lambda(t)$ if the dependency on time is obvious in the context.} is the Laplace noise calibrated to $\varepsilon$ as it has been described in Section \ref{sec:perturbation}. There are (at least) 3 possible approaches to do this which are detailed as follows.

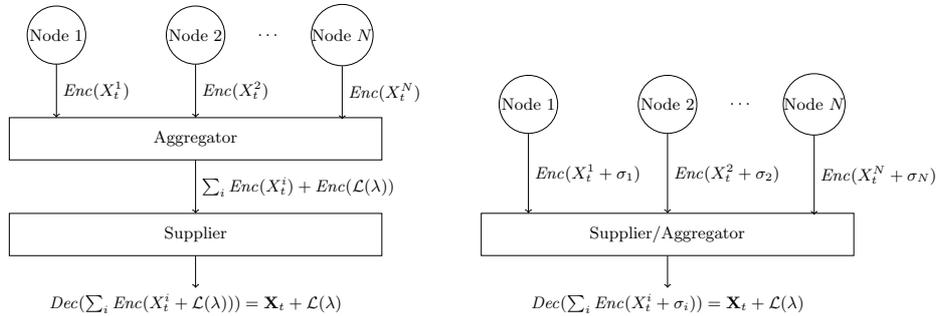
\begin{figure*}[ht]
\centering
\subfigure[Centralized approach: aggregation with trusted aggregator.]
{
\label{fig:centralized}
\scalebox{0.7} {
\begin{footnotesize}
\begin{tikzpicture}
[user/.style={circle,draw, inner sep=0pt,minimum size=1.1cm},
supp/.style={rectangle,draw, inner sep=0pt,minimum width=7cm, minimum height=0.8cm},auto]

\node[user] (n1) at (0,4) {Node 1};
\node[user] (n2) [right=1.5cm of n1.east] {Node 2};
\node [right=5mm of n2.east] (others) {$\ldots$};
\node[user] (nN) [right=5mm of others.east] {Node $N$};

\node (A1) [below=of n1] {};
\node[supp] (A2) [below=of n2] {Aggregator};
\node (A3) [below=of n2] {};
\node (A4) [below=of nN] {};

\node[supp] (S) [below=1cm of A2] {Supplier};
\node (text) [below=0.6 cm of S] {$\mathit{Dec}(\sum_{i}\mathit{Enc}(X_{t}^i + \mathcal{L}(\lambda))) = \mathbf{X}_{t} + \mathcal{L}(\lambda)$};

\path[->] (A2) edge node {$\sum_{i}\mathit{Enc}(X_{t}^{i}) + \mathit{Enc}(\mathcal{L}(\lambda))$} (S)
(n1) edge node {$\mathit{Enc}(X_{t}^{1})$} (A1)
(n2) edge node {$\mathit{Enc}(X_{t}^{2})$} (A2)
(nN) edge node {$\mathit{Enc}(X_{t}^{N})$} (A4)
(S) edge (text);
\end{tikzpicture}
\end{footnotesize}
}
}
\subfigure[Our approach: aggregation without trusted entity. If $\sigma_{i} = \mathcal{G}_1(N, \lambda) + \mathcal{G}_2(N, \lambda)$, where $\mathcal{G}_1$, $\mathcal{G}_2$ are i.i.d gamma noise, then $\sum_{i=1}^N\sigma_{i} = \mathcal{L}(\lambda)$.]
{
\label{fig:our_approach}
\scalebox{0.7} {
\begin{footnotesize}
\begin{tikzpicture}
[user/.style={circle,draw, inner sep=0pt,minimum size=1.1cm},
supp/.style={rectangle,draw, inner sep=0pt,minimum width=7cm, minimum height=0.8cm},auto]

\node[user] (n1) at (0,4) {Node 1};
\node[user] (n2) [right=1.5cm of n1.east] {Node 2};
\node [right=5mm of n2.east] (others) {$\ldots$};
\node[user] (nN) [right=5mm of others.east] {Node $N$};

\node (A1) [below=1.5 cm of n1] {};
\node[supp] (A2) [below=1.5 cm of n2] {Supplier/Aggregator};
\node (A3) [below=1.5 cm of n2] {};
\node (A4) [below=1.5 cm of nN] {};
\node (text) [below=0.6 cm of A2] {$\mathit{Dec}(\sum_{i}\mathit{Enc}(X_{t}^i + \sigma_{i})) = \mathbf{X}_{t} + \mathcal{L}(\lambda)$}; 



\path[->] (n1) edge node {$\mathit{Enc}(X_{t}^{1} + \sigma_{1})$} (A1)
(n2) edge node {$\mathit{Enc}(X_{t}^{2} + \sigma_{2})$} (A2)
(nN) edge node {$\mathit{Enc}(X_{t}^{N} + \sigma_{N})$} (A4)
(A2) edge  (text)
;
\end{tikzpicture}
\end{footnotesize}
}
}
\caption{Aggregating measurements while guaranteeing differential privacy.}
\end{figure*}

\subsection{Fully decentralized approach (without aggregator)}
\label{sec:trusted_appoach}
Our first attempt is that each user adds some noise to its own measurement, where the noise
is drawn from a Laplace distribution. In particular, each node $i$ sends the value of $X_{t}^{i} + \mathcal{L}(\lambda)$ directly to the supplier in time $t$.  
It is easy to see that $\varepsilon$ is guaranteed to all users, but in fact
the final noise added to the aggregate (i.e., $\sum_{i=1}^N\mathcal{L}(\lambda)$) is $N$ times larger than $\mathcal{L}(\lambda)$, and hence, the error is $\mu(t) = \frac{1}{\mathbf{X}_t + 1} \mathbb{E}|\sum_{i=1}^N \mathcal{L}(\lambda)| = \frac{N \cdot \lambda}{\mathbf{X}_t + 1}$. 

\subsection{Aggregation with a trusted aggregator}
Our second attempt can be to aggregate the measurements of some users, and send the perturbed
aggregate to the supplier. In particular, nodes are grouped into $N$ sized clusters and each node of a cluster sends its measurement $X_{t}^{i}$ to the (trusted) cluster aggregator, that is a trusted entity different from the supplier. 
The aggregator computes $\mathbf{X}_{t} = \sum_{i=1}^{N} X_{t}^{i}$  and obtains $\hat{\mathbf{X}}_{t} = \mathbf{X}_{t} + \mathcal{L}(\lambda)$ by adding noise to the aggregate. This perturbed aggregate is then sent to the supplier as it is illustrated in Figure \ref{fig:centralized}.

The utility of this approach is better than in the previous case, as the noise is only added to the sum and not to each measurement $X_{t}^{i}$. Formally, $\mu(t) = \frac{1}{\mathbf{X}_{t}+1}\mathbb{E}|\mathcal{L}(\lambda)| = \frac{\lambda}{\mathbf{X}_{t}+1}$. Similary, 
$\delta(t) = \frac{1}{\mathbf{X}_{t}+1} \cdot \sqrt{\mathbb{E}|\mathcal{L}(\lambda)|^2 - (\mathbb{E}|\mathcal{L}(\lambda)|)^2} = \frac{\lambda}{\mathbf{X}_{t}+1}$. 

However, the main drawback of this approach is that the aggregator must be 
fully trusted since it receives each individual measurement from the users.
This can make this scheme impractical if there is no such trusted entity.

\subsection{Our approach: aggregation without trusted entity}

Although the previous scheme is differential private, it works only if the aggregator is trustworthy and faithfully adds the noise to the measurement. In particular, the scheme will not be secure if the aggregator omits to add the noise. 

Our scheme, instead, does not rely on any centralized aggregator. The noise is added by each smart meter on their individual data and encrypted in such a way that the aggregator can only compute the (noisy) aggregate. Note that with our approach the aggregator and the supplier do need to be separate entities. The supplier can even play the role of the aggregator, as the encryption prevents it to access individual measurements, and the distributed generation of the noise ensures that it cannot manipulate the noise.

Our proposal is composed of 2 main steps: distributed generation of the Laplacian noise and encryption of individual measurements. These 2 steps are described in the remainder of this section.

\subsubsection{Distributed noise generation: a new approach}
In our proposal, the Laplacian noise is generated in a fully distributed way as is illustrated in Figure \ref{fig:our_approach}.  We use
the following lemma that states that the Laplace distribution is divisible and be constructed as the sum of i.i.d. gamma distributions. As this divisibility is infinite, it works for arbitrary number of users.   

\begin{lemma}[Divisibility of Laplace distribution \cite{kotz01book}]
\label{lem:gamma}
Let $\mathcal{L}(\lambda)$ denote a random variable which has a Laplace distribution with PDF $f(x,\lambda) = \frac{1}{2\lambda} e^{\frac{|x|}{\lambda}}$. Then the distribution of $\mathcal{L}(\lambda)$ is infinitely divisible. Furthermore, for every integer $n \geq 1$, $\mathcal{L}(\lambda) = \sum_{i=1}^{n} [\mathcal{G}_{1}(n, \lambda) - \mathcal{G}_{2}(n, \lambda)]$, where $\mathcal{G}_{1}(n, \lambda)$ and $\mathcal{G}_{2}(n, \lambda)$ are i.i.d. random variables having gamma distribution with PDF $g(x,n,\lambda) = \frac{(1/\lambda)^{1/n}}{\Gamma(1/n)} x^{\frac{1}{n}-1}e^{-x/\lambda}$ where $x \geq 0$.  
\end{lemma}

The lemma comes from the fact that $\mathcal{L}(\lambda)$ can be represented as the difference of two i.i.d exponential random variables with rate parameter $1/\lambda$. Moreover, $\sum_{i=1}^{n} \mathcal{G}_{1}(n, \lambda) - \sum_{i=1}^{n} \mathcal{G}_{2}(n, \lambda) = \mathcal{G}_{1}(1/\sum_{i=1}^{n}\frac{1}{n}, \lambda) - \mathcal{G}_{2}(1/\sum_{i=1}^{n}\frac{1}{n}, \lambda)= \mathcal{G}_{1}(1, \lambda) - \mathcal{G}_{2}(1, \lambda)$ due to the summation property of the gamma distribution\footnote{The sum of i.i.d. gamma random variables follows gamma distribution (i.e., $\sum_{i=1}^{n} \mathcal{G}(k_{i}, \lambda) = \mathcal{G}(1/\sum_{i=1}^{n}\frac{1}{k_{i}}, \lambda)$).}. Here, $\mathcal{G}_{1}(1, \lambda)$ and $\mathcal{G}_{2}(1, \lambda)$ are i.i.d exponential random variable with rate parameter $1/\lambda$ which completes the argument. 

Our distributed sanitization algorithm is simple; user $i$ calculates value $\hat{X}_{t}^{i} = X_{t}^{i} + \mathcal{G}_{1}(N, \lambda) - \mathcal{G}_{2}(N, \lambda)$ in slot $t$ and sends it to the aggregator, where $\mathcal{G}_{1}(N, \lambda)$ and $\mathcal{G}_{2}(N, \lambda)$ denote two random values independently drawn from the same gamma distribution. Now, if the aggregator sums up all values received
from the $N$ users of a cluster, then $\sum_{i=1}^{N} \hat{X}_{t}^{i} = \sum_{i=1}^{N} X_{t}^{i} + \sum_{i=1}^{N} [\mathcal{G}_{1}(N, \lambda) - \mathcal{G}_{2}(N, \lambda)] = \mathbf{X}_{t} + \mathcal{L}(\lambda)$ based on Lemma \ref{lem:gamma}.  

The utility of our distributed scheme is defined as $\mu(t) = \frac{1}{\mathbf{X}_{t}+1}\mathbb{E}|\mathbf{X}_{t} - \mathbf{X}_{t} + \sum_{i=1}^{n} [\mathcal{G}_{1}(N, \lambda) - \mathcal{G}_{2}(N, \lambda)]| = \frac{\mathbb{E}|\mathcal{L}(\lambda)|}{\mathbf{X}_{t}+1} =  \frac{\lambda}{\mathbf{X}_{t}+1}$, and $\delta(t) = \frac{\lambda}{\mathbf{X}_{t}+1}$.

\subsubsection{Encryption} 
\label{sec:enc}

The previous step is not enough to guarantee
privacy as only the sum of the measurements (i.e., $\hat{\mathbf{X}}_{t}$) is differential private
but not the individual measurements. In particular, the aggregator has access to $\hat{X}_{t}^{i}$, 
and even if $\hat{X}_{t}^{i}$ is noisy, $\mathcal{G}_{1}(N, \lambda) - \mathcal{G}_{2}(N, \lambda)$ is
usually insufficient to provide reasonable privacy for individual users if $N \gg 1$. This is illustrated in Figure \ref{fig:ill_noise}, where an individual's noisy and original measurements slightly differ.

\begin{figure*}[ht]
\centering
\subfigure[$X_{t}^{i}$]
{
\includegraphics[scale=0.45]{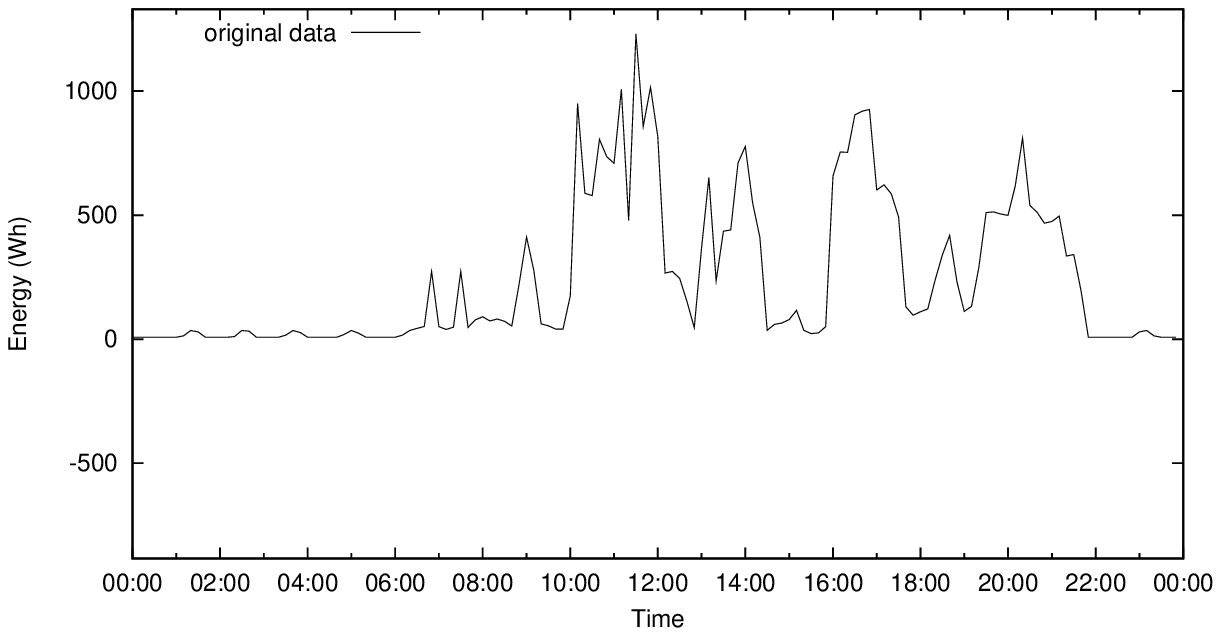}
}
\subfigure[$X_{t}^{i} + \mathcal{G}_{1}(N, \lambda) - \mathcal{G}_{2}(N, \lambda)$]
{
\includegraphics[scale=0.45]{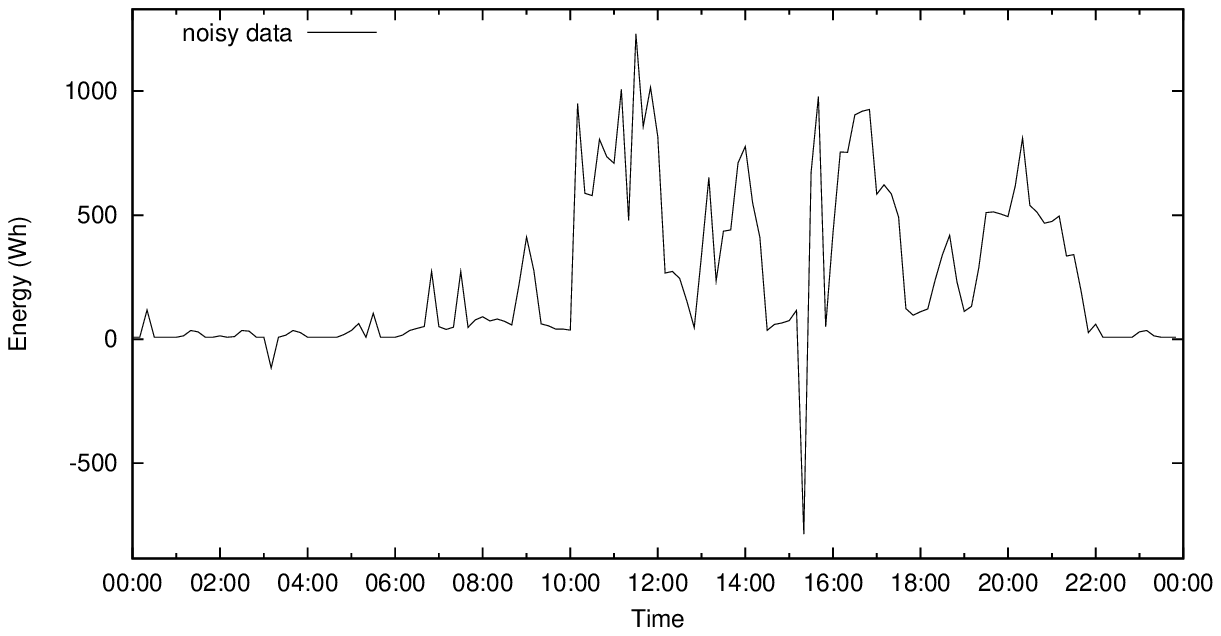}
}
\caption{\label{fig:ill_noise} The original and noisy measurements of user $i$, where the added noise is $\mathcal{G}_{1}(N, \lambda) - \mathcal{G}_{2}(N, \lambda)$ ($N=100$, $T_{p}$ is 10 min). }
\end{figure*}

To address this problem, each contribution
is encrypted using a modulo addition-based encryption scheme, inspired by \cite{cc05mobiquitous}, 
such that the aggregator
can only decrypt the sum of the individual values, and cannot access any of them. 
In particular, let $k_i$ denote a random key generated by user $i$ inside a cluster such that $\sum_{i=1}^{N} k_i = 0$, and $k_i$ is not known to the aggregator. Furthermore, $Enc()$ denotes a probabilistic encryption scheme such that 
$Enc(p,k,m) = p + k \mod m$, where $p$ is the plaintext, $k$ is the encryption key, and $m$ is a large integer.
The adversary cannot decrypt any $\mathit{Enc}(\hat{X}_t^{i},k_i,m)$, since it does not know $k_i$, but it can easily retrieve the noisy sum by adding the encrypted noisy measurements of all users; $\sum_{i=1}^{N} Enc(\hat{X}_t^{i},k_i,m) = \sum_{i=1}^N \hat{X}_t^i + \sum_{i=1}^N k_i = \sum_{i=1}^N \hat{X}_t^i \mod m$. If $z=\max_{i,t}(\hat{X}_t^i)$ then $m$ should be selected as $m=2^{\lceil \log_2(z \cdot N) \rceil}$ \cite{cc05mobiquitous}. The generation of $k_i$ is described in Section \ref{sec:processing}.

\section{Protocol description}

\subsection{System setup}
\label{sec:setup}

In our scheme, nodes are grouped into clusters of size $N$, where $N$ is a
parameter. The protocol requires the establishment of pairwise keys between each
pair of nodes inside a cluster that can be done by using traditional Diffie-Hellman key 
exchange as follows. When a node $v_{i}$ is installed, it provides a self-signed DH component and its certificate to the supplier. 
Once all the nodes of a cluster are installed, or a new node is deployed, the supplier broadcasts the certificates and public DH components of all nodes. Finally, each node $v_i$ of the cluster can compute a pairwise key $K_{i,j}$ shared with any other node $v_j$ in the networks. Note that no communication is required between $v_i$ and $v_j$.

\subsection{Smart meter processing}
\label{sec:processing}
Each node $v_i$ sends at time $t$ its periodic measurement, $X^i_t$, to the supplier as follows:

\begin{description}
\item[\textbf{Phase 1 (Data sanitization):}]  
Node $v_i$ calculates value 
$\hat{X}_{t}^{i} = X_{t}^{i} + \mathcal{G}_{1}(N, \lambda) - \mathcal{G}_{2}(N, \lambda)$, 
where $\mathcal{G}_{1}(N, \lambda)$ and $\mathcal{G}_{2}(N, \lambda)$ denote two random 
values independently drawn from the same gamma distribution and $N$ is the cluster size.

\item[\textbf{Phase 2 (Data encryption):}]  
Each noisy data $\hat{X}_{t}^{i}$ is then encrypted into  $\mathit{Enc}(\hat{X}_{t}^{i})$ using the modulo addition-based encryption scheme detailed in Section \ref{sec:enc}. The following extension is then applied to generate the encryption keys:
Each node, $v_i$, selects $\ell$ other nodes randomly, such that if $v_i$ selects $v_j$, then $v_j$ also selects $v_i$. 
Afterwards, both nodes generate a common dummy key $k$ from their pairwise key $K_{i,j}$; $v_i$ adds $k$ 
to $\mathit{Enc}(\hat{X}_{t}^{i})$ and $v_j$ adds $-k$ to $\mathit{Enc}(\hat{X}_{t}^{j})$. 
As a result, the aggregator cannot decrypt the individual ciphertexts (it does not know the dummy key $k$). However, it adds all
the ciphertexts of a given cluster, the dummy keys cancel out and it retrieves the encrypted sum of the (noisy) contributions.
The more formal description is as follows:

\begin{enumerate} 
\item node $v_i$ selects some nodes of the cluster randomly (we call them participating nodes) using a secure pseudo random function (PRF) such that if $v_i$ selects $v_j$, then $v_j$ also selects $v_i$. 
In particular, $v_i$ selects $v_j$ if mapping $\mathit{PRF}(K_{i,j}, r_{1})$ to a value between 0 and 1 is less or equal than $\frac{w}{N-1}$, where $r_{1}$ is a public value changing in each slot.
We denote by $\ell$ the number of selected participating nodes, and $\mathsf{ind}_i[j]$ (for $j=1, \ldots,\ell$)  denotes the index of the $\ell$ nodes selected by node $v_i$. Note that, for the supplier, the probability that $v_i$ selects $v_j$ is $\frac{w}{N-1}$ as it does not know $K_{i,j}$. The expected value of $\ell$ is $w$.

\item $v_i$ computes for each of its $\ell$ participating nodes a {\em dummy key}. A dummy key between $v_i$ and $v_j$ is defined as $\mathsf{dkey}_{i,j}= (i-j)/|i-j|\cdot\mathit{PRF}(K_{i,j}, r_{2})$, where $K_{i,j}$ is 
the key shared by $v_i$ and $v_j$, and $r_{2} \neq r_{1}$ is public value changing in each slot.  Note that $\mathsf{dkey}_{i,j}= -\mathsf{dkey}_{j,i}$.

\item $v_i$ then computes 
$\mathit{Enc}(\hat{X}_{t}^{i}) = \hat{X}_{t}^{i}  + K'_i + \sum^\ell_{j=1} \mathsf{dkey}_{i,\mathsf{ind}_i[j]} \pmod{m}$, where $K'_i$ is the keystream shared by $v_i$ and the aggregator which can be established using the DH protocol as above, and $m$ is a large integer (see \cite{cc05mobiquitous}). Note that $m$ must be larger than the sum of all contributions (i.e., final aggregate) plus the Laplacian noise.\footnote{Note that the noise is a random value from an infinite domain and this sum might be larger than $m$. However, choosing sufficiently large $m$, the probability that the sum exceeds $m$ can be made arbitrary small due to the exponential tail of the Laplace distribution.}  

Note that $\hat{X}_{t}^{i}$ is encrypted multiple times: it is first encrypted with the keystream $K'_i$
and then with several dummy keys. $K'_i$ is needed to ensure confidentiality between a user and the aggregator. 
The dummy keys are needed to prevent the aggregator (supplier) from retrieving $\hat{X}_{t}^{i}$.
\item $\mathit{Enc}(\hat{X}_{t}^{i})$ is sent to the aggregator (supplier).
\end{enumerate}
\end{description}

\subsection{Supplier processing}

\begin{description}
\item[Phase 1 (Data aggregation):] At each slot, the supplier aggregates the $N$ measurements received from the cluster smart meters by summing them, and obtains $\sum^N_{i=1} \mathit{Enc}(X_{t}^{i})$. In particular,
$$\mathit{Enc}(\hat{\mathbf{X}}_{t}) = \sum^N_{i=1} (\hat{X}_{t}^{i} + K'_i)  + \sum^N_{i=1}\sum^\ell_{j=1} \mathsf{dkey}_{i,\mathsf{ind}_i[j]} \pmod{m}$$
where $\sum^N_{i=1}\sum^\ell_{j=1} \mathsf{dkey}_{i,\mathsf{ind}_i[j]}=0$ because $\mathsf{dkey}_{i,j}=-\mathsf{dkey}_{j,i}$. Hence, 
$$
\mathit{Enc}(\hat{\mathbf{X}}_{t}) = \sum^N_{i=1} (\hat{X}_{t}^{i} + K'_i) = \sum^N_{i=1} \mathit{Enc}(\hat{X}_{t}^{i})
$$

\item[Phase 2 (Data decryption):] The aggregator then decrypts the aggregated value by subtracting the sum of the node's keystream, and retrieves the sum of the noisy measures:
$$
\sum^N_{i=1} \mathit{Enc}(\hat{X}_{t}^{i}) - \sum^N_{i=1} K'_i = \sum^N_{i=1} \hat{X}_{t}^{i}  \pmod{m}
$$
where $\sum^N_{i=1} \hat{X}_{t}^{i} = \sum^N_{i=1} X_{t}^{i} + \sum^N_{i=1} \mathcal{G}_{1}(N, \lambda) - \sum^N_{i=1} \mathcal{G}_{2}(N, \lambda) = \sum^N_{i=1} X_{t}^{i} + \mathcal{L}(\lambda)$
based on Lemma \ref{lem:gamma}.
\end{description}
The main idea of the scheme is that the aggregator is not able to decrypt the
individual encrypted values because it does not know the dummy keys.
However, by adding the different encrypted contributions, dummy keys
cancel each other and the aggregator can retrieve the sum of the plaintext. 
The resulting plaintext is then the perturbed sums of the measurements, where the
noise ensures the differential privacy of each user.

\paragraph{Complexity:} Let $b$ denote the size of the pairwise keys (i.e., $K_{i,j}$). Our scheme has $O(N \cdot b)$ storage complexity, as each node needs to store $\ell \leq N$ pairwise keys. The computational overhead is dominated by the encryption and the key generation complexity. The encryption is composed of $\ell \leq N$ modular addition of $\log_2 m$ bits long integers, while the key generation needs the same number of PRF executions. This results in a complexity of $O(N \cdot (\log_2 m + c(b)))$, where $c(b)$ is the complexity of the applied PRF function. \footnote{For instance, if $\log_2 m=32$ bits (which should be sufficient in our application), $b=128$, and $N=1000$, a node needs to store 16 Kb of key data and perform  maximum 1000 additions along with 1000 subtractions (for modular reduction) on 32 bits long integers, and maximum 1000 PRF executions. This overhead should be negligible even on constrained embedded devices.}

\section{Adding robustness}
\label{sec:robust}

We have assumed so far that all the $N$ nodes of a cluster
participated in the protocol. However, it might happen that, for
several different reasons (e.g., node or communication failures) some
nodes are not able to participate in each epoch. This would have two effects: first, 
security will be reduced since the sum of the noise added by each node will not be 
equivalent to $\mathcal{L}(\lambda)$. Hence, differential privacy may not be guaranteed.
Second, the aggregator will not be able to decrypt the aggregated value since the sum of 
the dummy keys will not cancel out.

In this section, we extend our scheme to resist node failures. We propose
a scheme which resists the failure of up to $M$ out of $N$ nodes, 
where $M$ is a configuration parameter. We will study later the
impact of the value $M$ on the scheme performance.

\subsection{Sanitization phase extension}

In order to resist the failure of $M$ nodes, each node should add the following
noise to their individual measurement: $\mathcal{G}_{1}(N-M, \lambda) - \mathcal{G}_{2}(N-M, \lambda)$. Note that
$\sum_{i=1}^{N-M} [\mathcal{G}_{1}(N-M, \lambda) - \mathcal{G}_{2}(N-M, \lambda)] = \mathcal{L}(\lambda)$. Therefore, this sanitization algorithm remains differential private, if at least $N-M$ nodes participate in the protocol. Note that in that case each node adds extra noise to the aggregate in order to ensure differential privacy even if fewer than $M$ nodes fail to send their noise share to the aggregator. 

\subsection{Encryption phase extension}
\subsubsection{A simple approach}

As described previously, all the dummy keys cancel out at the aggregator.
However, this is not the case if not all the nodes participate in the protocol. In order to resist the failure of nodes, one can extend the encryption scheme with an additional round where the 
aggregator asks the participating nodes of non-responding nodes for the missing dummy keys:
\begin{enumerate}
\item Once the aggregator received all contributions, it broadcasts the ids of the non-responding nodes. 
\item Upon the reception of this message, each node $v_{i}$ verifies whether any of the ids in the broadcast message are in its participating node list (i.e., it can be found in $\mathsf{ind}_{i}$). For each of such id, the node sends the corresponding dummy key to the aggregator.
\item The aggregator then subtracts all received dummy keys from $\mathit{Enc}(\hat{\mathbf{X}}_{t})$ and retrieves  $\sum^N_{i=1} (\hat{X}_{t}^{i} + K'_i)$ which can be decrypted.
\end{enumerate}

This approach has a severe problem: if the aggregator is untrusted, it can easily retrieve the measurement of a $v_{i}$: broadcasting its id in Step 2, the participating nodes of $v_i$ reply with the dummy keys of $v_i$ which can be removed from $\mathit{Enc}(\hat{X}_{t}^{i})$. 

\subsubsection{Our proposal} 
\label{sec:advanced} 

In this approach, each node adds a secret random value to its encrypted value before releasing it in the first round. This is needed to prevent the adversary to recover the noisy measurement through combining different messages of the nodes. Then, in the second round when the aggregator asks for the missing dummy keys, every node reveals its random keys along with the missing dummy keys that it knows:
\begin{enumerate}
\item Each node $v_{i}$ sends 
$
\mathit{Enc}(\hat{X}_{t}^{i}) = \hat{X}_{t}^{i} + K'_{i} + \sum^\ell_{j=1} \mathsf{dkey}_{i,\mathsf{ind}_i[j]} + C_{i} \pmod{m}
$
where $C_{i}$ is the secret random key of $v_{i}$ generated randomly in each round.

\item After receiving all measurements, the aggregator asks all nodes for their random keys and the missing dummy keys through broadcasting the id of the non-responding nodes. 

\item Each node $v_{i}$ verifies whether any ids in this broadcast message are in its participating node list, where the set of the corresponding participating nodes is denoted by $S$. Then, $v_i$ replies with $\sum_{j \in S} \mathsf{dkey}_{i,\mathsf{ind}_i[j]} + C_{i} \pmod{m}$. 

\item The aggregator subtracts all received values from $\sum_{i=1}^N \mathit{Enc}(\hat{X}_{t}^i)$ which results in $\sum^N_{i=1} (\hat{X}_{t}^{i} + K'_i)$, as the random keys as well as the dummy keys cancel out. 
\end{enumerate} 

Note that as the supplier does not know the random keys, it cannot remove them from any messages but only from the final aggregate; adding each node's response to the aggregate all the dummy keys and secret random keys cancel out and the supplier obtains $\hat{\mathbf{X}}_{t}$. 
Although
the supplier can still recover $\hat{X}_t^i$ if it knows $v_i$'s participating nodes (the supplier simply asks for all the dummy keys of $v_i$ in Step 2 and subtracts $v_i$'s response in Step 4 from $\mathit{Enc}(\hat{X}_t^i)$), we will show later that this probability can be made practically small by adjusting $w$ and $N$ correctly.   

Note that the protocol fails if, for some reasons, a node does not send its random key to the aggregator (as only the node itself knows its random key, it cannot be reconstructed by other parties). However, it is very unlikely that a node between the two rounds fails, and an underlying reliable transport protocol helps to overcome communication errors. 

Finally, also note that this random key approach always requires two rounds of communication (even if the aggregator receives all encrypted values correctly in the first round), as the random keys are needed to be removed from $\mathit{Enc}(\hat{\mathbf{X}}_{t})$ in the second round.

\subsection{Utility evaluation}
\label{sec:extra_noise}

If all $N$ nodes participate in the protocol, the added noise will
be larger than $\mathcal{L}(\lambda)$ which is needed to ensure differential privacy. 
In particular, $\sum_{i=1}^{N} [\mathcal{G}_{1}(N-M, \lambda) - \mathcal{G}_{2}(N-M, \lambda)] = 
\mathcal{L}(\lambda) + \sum_{i=1}^{M} [\mathcal{G}_{1}(N-M, \lambda) - \mathcal{G}_{2}(N-M, \lambda)]$, where
the last summand is the extra noise needed to tolerate the failure of maximum $M$ nodes. Clearly, this extra noise increases the error if all $N$ nodes operate correctly and add their noise shares faithfully. In what follows, we calculate the error and its standard deviation if we add this extra noise to the aggregate.

\begin{theorem}[Utility]
\label{thm:utility}
Let $\alpha = M/N$ and $\alpha < 1$. Then,
$$
\mu(t) \leq \frac{2}{B(1/2,\frac{1}{1-\alpha})} \cdot \frac{\lambda(t)}{\mathbf{X}_{t}+1}
$$ 
and 
$$ 
\sigma(t) \leq \sqrt{\left( \frac{2}{1- \alpha} - \frac{4}{B(1/2,\frac{1}{1-\alpha})^2} \right)} \cdot \frac{\lambda(t)}{\mathbf{X}_{t}+1}
$$
where $B(x,y) = \frac{\Gamma(x)\Gamma(y)}{\Gamma(x+y)}$ is the beta function.
\end{theorem}

The derivation can be found in Appendix \ref{app:proof}.  
Based on Theorem \ref{thm:utility}, 
$
\sigma(t) =  \mu(t) \cdot \left(\frac{2}{B(1/2,\frac{1}{1-\alpha})}\right)^{-1} \cdot \sqrt{\left( \frac{2}{1- \alpha} - \frac{4}{B(1/2,\frac{1}{1-\alpha})^2} \right)}
$. It is easy to check that $\sigma(t)$ is always less or equal than $\mu(t)$. In particular, if $\alpha=0$ (there are no malicious nodes and node failures), then $\sigma(t) = \mu(t)$. If $\alpha > 0$ then $\sigma(t) < \mu(t)$ but $\sigma(t) \approx \mu(t)$. 

%


\section{Security Analysis}

\subsection{Deploying malicious nodes}

In the proposed scheme, each measurement is perturbed and encrypted. 
Therefore, a honest-but-curious attacker cannot
gain any information (up to $\varepsilon$) about individual measurements in any slot. This is guaranteed by the encryption scheme and the added noise. 

However, a DN adversary (see Section \ref{sec:adv_model}), which deploys $T$ malicious nodes, may be able to:
\begin{itemize}
\item reduce the noise level by limiting (or omitting) the gamma noise added by
malicious nodes. As a result, the sum of the noise shares will not 
equal to the Laplacian noise which can decrease the privacy of users.
However, recall that, due to the robustness property of our scheme detailed in Section \ref{sec:robust}, we
add extra noise to tolerate $M$ node failures. Adding extra noise calibrated to $M+T$ is sufficient to tolerate this type of attack.

\item decrypt $\mathit{Enc}(\hat{X}_{t}^{i})$ of a node $v_i$ and retrieve the perturbed data. 
As individual data is only weakly noised, the attacker might infer 
some information from them, and therefore, compromise privacy. 
However, the encryption scheme that we used is provably 
secure  \cite{cc05mobiquitous}, and nodes are assumed to be tamper-resistant. 
Thus, the only way to break privacy is to retrieve the dummy keys of $v_i$.
Because the participating nodes are selected randomly for each message, 
this can only be achieved if {\bf all} participating nodes of $v_i$ are
malicious {\bf and} the supplier is also malicious (i.e., the adversary knows $K'_{i}$). This happens if $v_i$ does not select any honest participating node that has a probability of $(1-\frac{w}{N-1})^{N-T-1}$. For instance, it is easy to check that if $N = 100$ and 50\% of the nodes are malicious (which anyway should be a quite strong assumption), then setting $w$ to 30 results in a success probability of $1.8 \cdot 10^{-8}$. This means that if an epoch is 5 min long, then the adversary will compromise 1 measurement during 458 years in average.

Finally, also note that this is the success probability of the adversary in a single slot. This means that a supplier that succeeds the previous attack only gets a single (noisy) measurement of 
the customer (corresponding to a single epoch). As a node selects different participating nodes in each slot, 
the probability that the adversary gets $k$ different measurements of the node is $(1-\frac{w}{N-1})^{k(N-T-1)}$, which is even smaller.
\end{itemize}

\subsection{Lying supplier}

\subsubsection*{Lying about non-responding nodes}

In addition to deploying malicious (fake) nodes, a malicious supplier can
lie about the non-responding nodes. 
In order to recover $\hat{X}_{t}^{i}$, the supplier needs $\sum^\ell_{j=1} \mathsf{dkey}_{i,\mathsf{ind}_i[j]} + C_{i}$. The supplier has two options to retrieve this sum.
First, it might pretend that a node $v_i$ did not respond in the first round,
and asks for $v_i$'s dummy keys to its participating nodes. At the same time, the supplier claims to $v_i$ that
its participating nodes are responding. 
Hence, as described in Section \ref{sec:advanced}, the participating nodes of $v_i$ will disclose $v_i$'s
dummy keys and $v_i$ will disclose $C_i$. However, the random keys of $v_i$'s participating nodes prevent the supplier to retrieve $v_i$'s dummy keys from their messages. 

Second, the supplier can pretend that $v_i$'s participating nodes do not respond in the first round,
and asks $v_i$ for their dummy keys in the second round. 
In particular, there are three types of dummy keys: the first is shared with a malicious node, and hence, known to the supplier. The second is asked to $v_i$ by the supplier in the second round (the supplier pretends that these nodes are non-responding), and $v_i$ replies with the sum of $C_i$ and the requested keys. Finally, the rest is shared with honest participating nodes and they are not asked to $v_i$ in the second round. Apparently, if $v_i$ has at least one dummy key from the last group, its measurement cannot be recovered. This is because if $v_j$ is a participating honest node of $v_i$ and $\mathsf{dkey}_{i,\mathsf{ind}_i[j]}$ is not asked to $v_i$ in the second round, it could only be recovered from $v_j$'s messages. However, $v_j$ sends $C_{j} + \mathsf{dkey}_{i,\mathsf{ind}_i[j]}$, where $C_j$ is only known to $v_j$.    
 
Nevertheless, it might happen that $v_i$ does not have any third-type dummy key (i.e., the supplier asks $v_i$ for all the dummy keys shared with honest nodes). Then, the supplier can easily recover $v_i$'s measurement, since it knows $\sum^\ell_{j=1} \mathsf{dkey}_{i,\mathsf{ind}_i[j]} + C_i$ (they are malicious keys or provided by $v_i$). However, the supplier can only guess $v_i$'s participating nodes and target them randomly since $v_i$ also selects them randomly\footnote{Note that \emph{all} nodes send responses in the second round, and the randomness of $C_i$ ensures that the supplier cannot gain any knowledge about the participating nodes of any nodes.}. 
Assuming that the supplier can ask $v_{i}$ for maximum $M$ dummy keys in the second round, the probability that all participating nodes of $v_i$ are either malicious or specified as non-responding nodes by the supplier is less than $(1-\frac{w}{N-1})^{N-(T+M)-1}$. Using $\alpha= (T+M)/N$ and $\beta=w/N$, then $(1-\frac{w}{N-1})^{N-(T+M)-1} = (1-\frac{\beta}{1-N^{-1}})^{N(1-\alpha)-1}$.  
This probability is depicted in Figure \ref{fig:adv_succ} depending on $\alpha, \beta$ and $N$.

\begin{figure*}[ht]
\centering
\subfigure[$N$=100]
{
\includegraphics[scale=0.5]{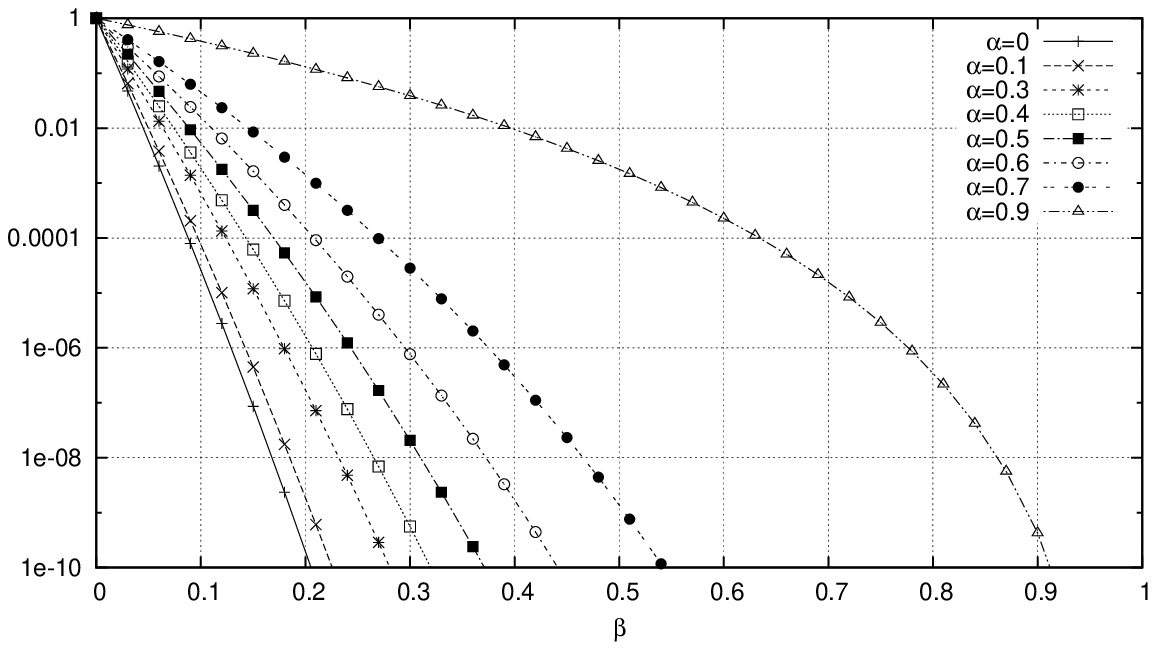}
}
\subfigure[$N$=300]
{
\includegraphics[scale=0.5]{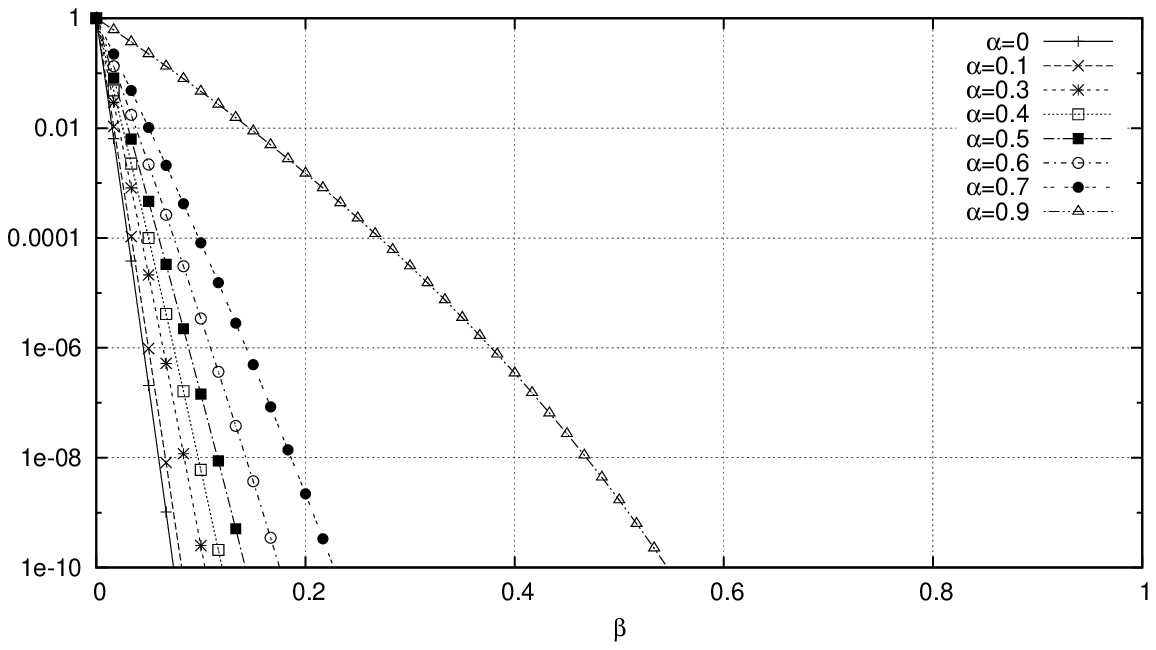}
}
\caption{Success probability of guessing participating nodes depending on $\beta$ and different values of $\alpha$ and $N$.}
\label{fig:adv_succ}
\end{figure*}

\subsubsection*{Lying about cluster size}

Another strategy for the supplier to compromise the privacy of users is to lie
about the cluster size. If the supplier pretends that the cluster
size $N'$ is larger than it really is (i.e., $N' > N$), the noise added by 
each node will be underestimated. In fact, each node will calibrate its 
gamma noise using $N'$ instead of $N$. As a result, the 
aggregated noise at the supplier will be smaller than necessary to 
guarantee sufficient differential privacy. 

In order to prevent this attack, a solution would be to set the cluster
size to a fixed value. For example, all clusters should have a size of 100.
Although simple and efficient, this solution is not flexible and might
not be applicable to all scenarios. Another option is that the supplier
publishes, together with the list of cluster nodes, a self-signed
certificate of each node of the cluster (containing a timestamp, the cluster id and the node 
information). That way, each node could verify 
the cluster size and get information about other member nodes.

\section{Simulation results}

\subsection{A high-resolution electricity trace simulator}

Due to the lack of high-resolution real world data, we implemented a domestic electricity demand model \cite{richardson10} that can generate one-minute resolution synthetic consumption data of different households\footnote{Available at \texttt{http://www.crysys.hu/\textasciitilde acs/misc/}}. It is an extended version of the simulator developed in \cite{richardson10}.  
The simulator includes 33 different appliances and implements a separate lighting model which takes into account the level of natural daylight depending on the month of the year. The number of residents in each household is randomly selected between 1 and 5. A trace is associated to a household and generated as follows: 
(1) A number of active persons is selected according to some distribution derived from real statistics. This number may vary as some members
can enter or leave the house. (2) A set of appliances is then selected and activated at different time of the day according to  another distribution,
which was also derived from real statistics. 
 
The input of the simulator is the number of households, the month of the day, and the type of the day (either a working or weekend day). The output is the power demand model (1-min profile) of all appliances in each household on the given day. Using this simulator, we generated 3000 electricity traces corresponding to different households on a working day in November, where the number of residents in each 
household was randomly selected between 1 and 5. Each trace was then sanitized according to our scheme.
The noise added in each slot (i.e., $\lambda(t)$) is set to the maximum consumption in the slot (i.e., $\lambda(t) = \max_{1\leq i \leq N}X_{t}^{i}$ where the maximum is taken on all users in the cluster). This amount of noise ensures $\varepsilon=1$ indistinguishability for individual measurements in all slots.
Although one can increase $\lambda(t)$ to get better privacy, the error will also increase. Note that the error $\mu_{\varepsilon'}(t)$ for other $\varepsilon' \neq \varepsilon$ values if $\mu_{\varepsilon}(t)$ is given is $\mu_{\varepsilon'}(t) = \frac{\varepsilon}{\varepsilon'} \cdot \mu_{\varepsilon}(t)$. We assume that $\lambda(t) = \max_{i}X_{t}^{i}$ is known a priori.

\subsection{Error according to the cluster size}  
    
The error introduced by our scheme depends on the cluster size $N$. In this section, we present how the error varies according to $N$.  

\subsubsection{Random clustering}

The most straightforward scheme to build $N$-sized clusters is to select $N$ users uniformly at random. The advantage of this approach is that users only need to send the noisy aggregate to the supplier. 
Figure \ref{fig:random_cluster} and \ref{fig:random_cluster_dev} show the average error value and its standard deviation, resp., depending on the size of the cluster. 
The average error of a given cluster size $N$ is the average of $\mathsf{mean}_{t}(\mu(t))$ of all 
$N$-sized clusters\footnote{In fact, the average error is approximated in Figure \ref{fig:random_cluster}:
we picked up 200 different clusters for each $N$, and plotted the average of their $\mathsf{mean}_{t}(\mu(t))$. 200 is chosen according to experimental analysis. Above 200, the average error do not change significantly.}. 
Obviously, higher $N$ causes smaller error. Furthermore, a high $\alpha$ results in larger noise added by each meters, as described in Section \ref{sec:extra_noise}, which also implies larger error.
Interestingly, increasing the sampling period (i.e., $T_p$) results in slight error decrease\footnote{This increase is less than 0.01 even if $N$ is small when the sampling period is changed from 5 min to 15 min.}, hence, we only considered 10 min sampling period. Otherwise noted explicitly, we assume 10 min sampling period in the sequel. 

\begin{figure*}[ht]
\centering
\subfigure[Average error]
{
\label{fig:random_cluster}
\includegraphics[scale=0.5]{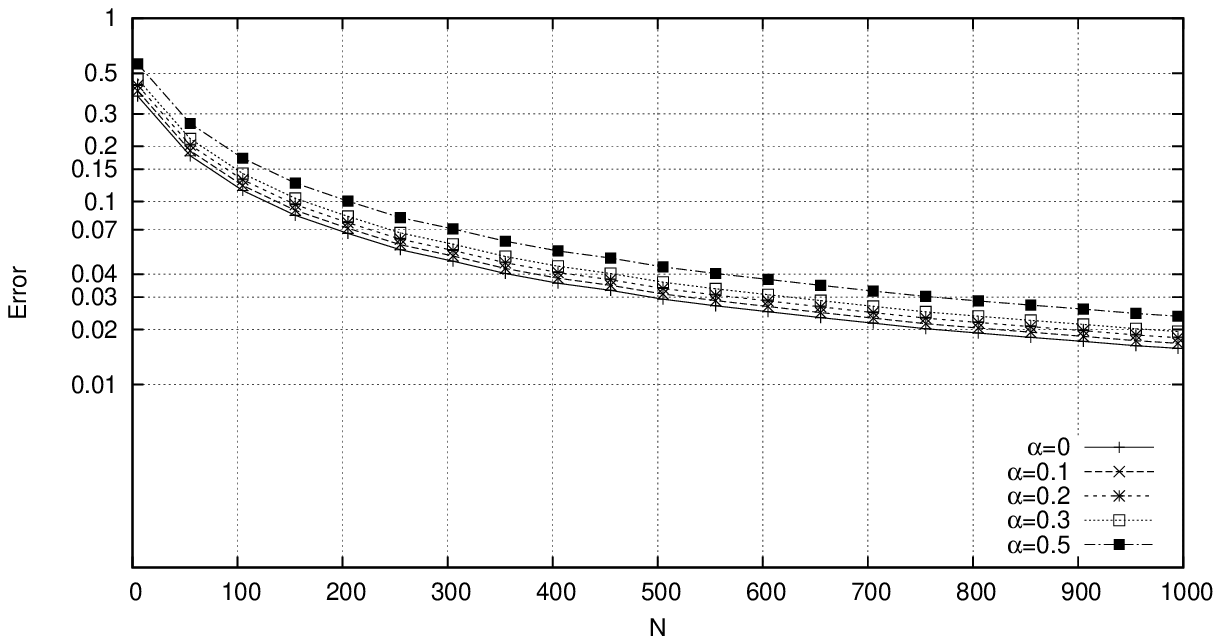}
}
\subfigure[Standard deviation of the average error]
{
\label{fig:random_cluster_dev}
\includegraphics[scale=0.5]{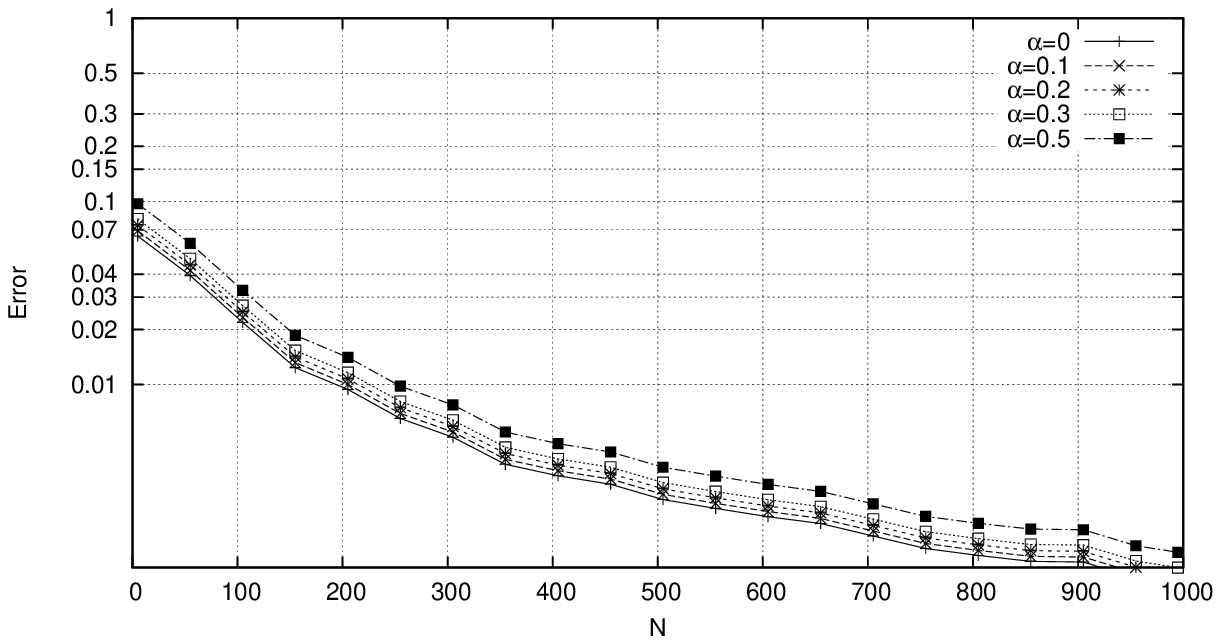}
}
\subfigure[Maximum error]
{
\label{fig:random_cluster_max}
\includegraphics[scale=0.5]{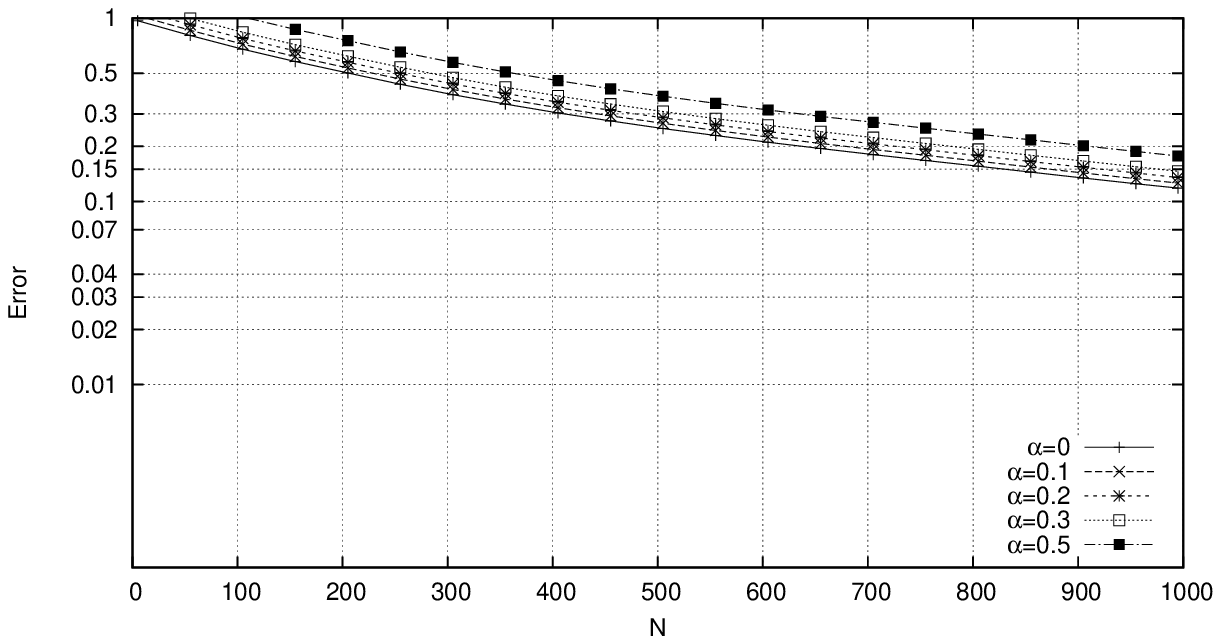}
}
\caption{\label{fig:random_cluster_main} The error depending on $N$ using random clustering. $T_{p}$ is 10 min.}

\end{figure*}

\subsubsection{Consumption based clustering}

As $\lambda(t)$ is set to the maximum consumption at $t$ inside a cluster, we could get lower error if the maximum consumption is close to the mean of the measurements within a cluster in every $t$. Hence, instead of randomly clustering users, a more clever approach is to cluster them based on the ``similarity'' of their consumption profiles. Intuitively, the measurements in similar profiles are close, and thus, the difference between the maximum consumption and the average should also be smaller than in a random cluster.

We measure profile similarity by the average daily consumption: the $N$-sized clusters are created by calculating daily consumption levels $\ell_1, \ell_2, \ldots, \ell_n$ (where $\ell_i \leq \ell_{i+1}$ for all $1 \leq i \leq n-1$) such that the number of users whose daily average is between $\ell_i$ and $\ell_{i+1}$ for all $i$ is exactly $N$. Then, all users being in the same level form a cluster.  
In contrast to random clustering, users need to provide the supplier with their daily averages which may leak some private information. However, this can also be derived from the (monthly) aggregate consumption of each user, which is generally revealed for the purpose of billing.

Figure \ref{fig:cons_cluster} and \ref{fig:cons_cluster_dev} show the average error and its deviation, resp., calculated identically to random clustering. Comparing Figure \ref{fig:cons_cluster_main} and \ref{fig:random_cluster_main}, consumption based clustering has lower error than the random one. The improvement varies up to 5\% depending on $N$.
For instance, while random clustering provides an average error of 0.13 with $N=100$ users in a cluster, consumption based clustering has 0.07. The difference decreases as $N$ increases. There are more significant  differences between the standard deviations and the worst cases: at lower values of $N$, the standard deviation of the average error in random clustering is almost twice as large as in consumption based clustering (Figure \ref{fig:cons_cluster_dev} and \ref{fig:random_cluster_dev}). To compute the worst case error, at a given $N$, the maximum error is computed in all slots, which is the highest cluster error that can occur in a slot with cluster size $N$. Then, the average of these maximum errors (the average is taken on all slots) are plotted in Figure \ref{fig:random_cluster_max} and \ref{fig:cons_cluster_max}. Apparently, the worst case error in random clustering is much higher than in consumption based clustering, as random clustering may put high and low consuming users into the same cluster.

\begin{figure*}[ht]
\centering
\subfigure[Average error]
{
\label{fig:cons_cluster}
\includegraphics[scale=0.5]{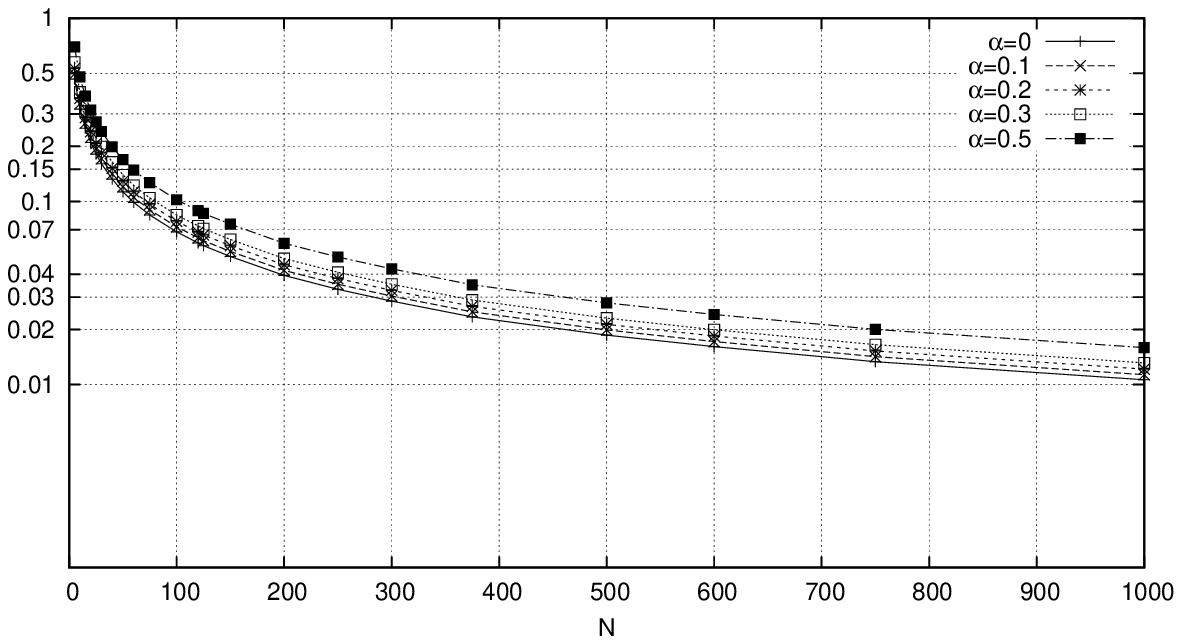}
}
\subfigure[Standard deviation of the average error]
{
\label{fig:cons_cluster_dev}
\includegraphics[scale=0.5]{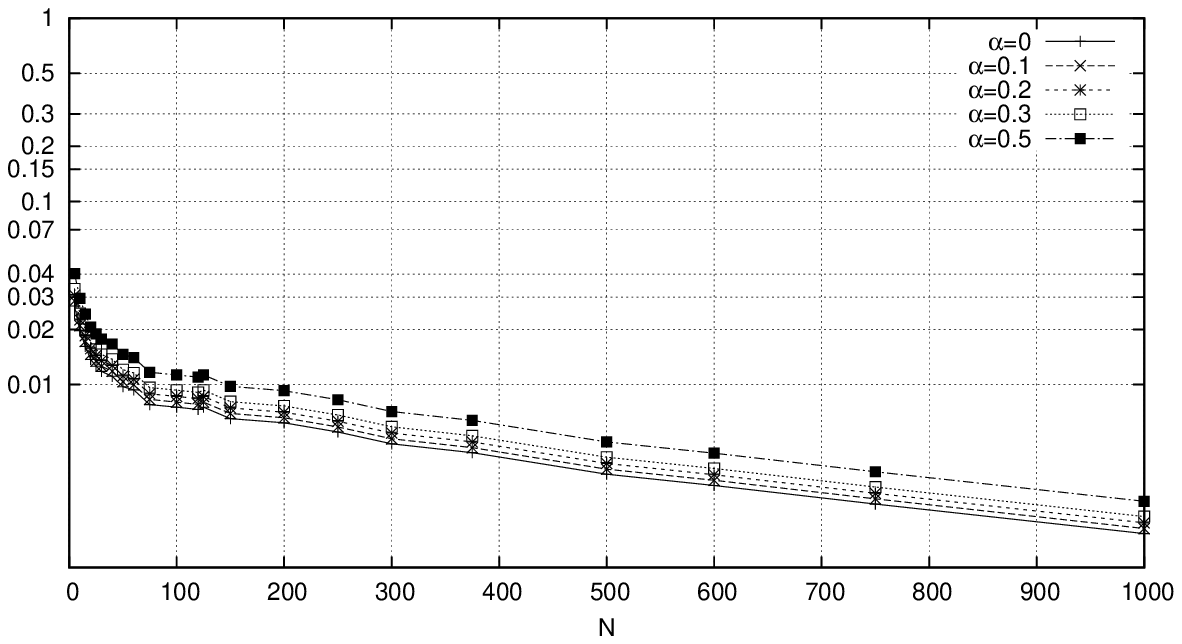}
}
\subfigure[Maximum error]
{
\label{fig:cons_cluster_max}
\includegraphics[scale=0.5]{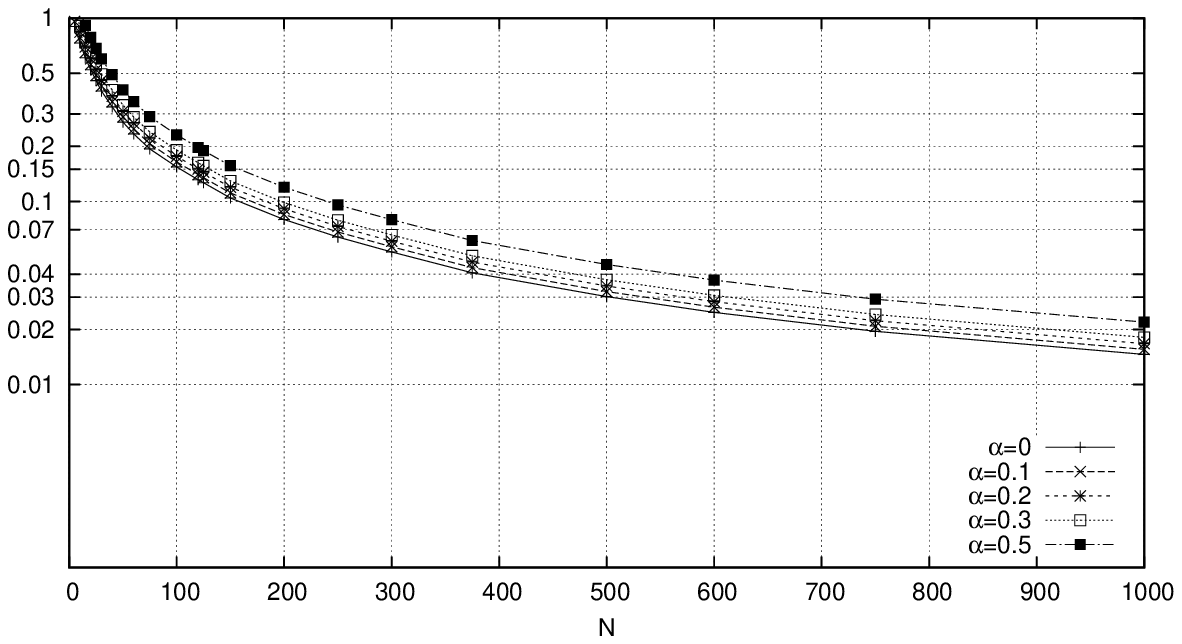}
}
\caption{The error depending on $N$ using consumption based clustering. $T_{p}$ is 10 min. \label{fig:cons_cluster_main} 
}
\end{figure*}

\subsection{Privacy over multiple slots}
So far, we have considered the privacy of individual slots, i.e. added noise to guarantee $\varepsilon=1$ privacy in each slot of size 10 minutes. However, a trace
is composed of several slots. For instance, if a user watches TV during multiple slots, we have guaranteed that an adversary cannot tell if the TV is watched in any particular slot (up to $\varepsilon=1$). However, by analysing $s$ consecutive slots corresponding to a given period,  it may be able to tell whether the TV was watched during that period (the privacy bound of this is $\varepsilon_s = \varepsilon \cdot s$ due to the composition property of differential privacy).   
Based on Theorem \ref{thm:histo}, we need to add noise $\lambda(t) = \sum_{i=1}^{s} \max_{i} X_{t}^{i}$ to \emph{each} aggregate to guarantee $\varepsilon_s=1$ bound in consecutive $s$ slots, which, of course, results in higher error than in the case of $s=1$ that we have assumed so far. Obviously, using the LPA technique, we cannot guarantee reasonably low error if $s$ increases, as the necessary noise $\lambda(t) = \sum_{i=1}^{s} \max_{i} X_{t}^{i}$ can be large. In order to keep the error $\lambda(t)/\sum_{i=1}^N X_{t}^i$ low while ensuring better privacy than $\varepsilon_s = s \cdot \varepsilon$, one can increase the number of users inside each cluster (i.e., $N$).

Figure \ref{fig:privacy_all_main} shows what average privacy of a user has, in our dataset, as a function of 
 the cluster size and value  $s$. 
As the cluster size increases, the privacy bound decreases (i.e. privacy increases). The reason is that when the cluster
size increases, the maximum consumption also increases with high probability. Since the noise is calibrated according to 
the maximum consumption within the cluster, it will be larger. This results in better privacy.

\begin{figure*}[ht]
\centering
\subfigure[All appliances]
{
\label{fig:privacy_all_main} 
\includegraphics[scale=0.5]{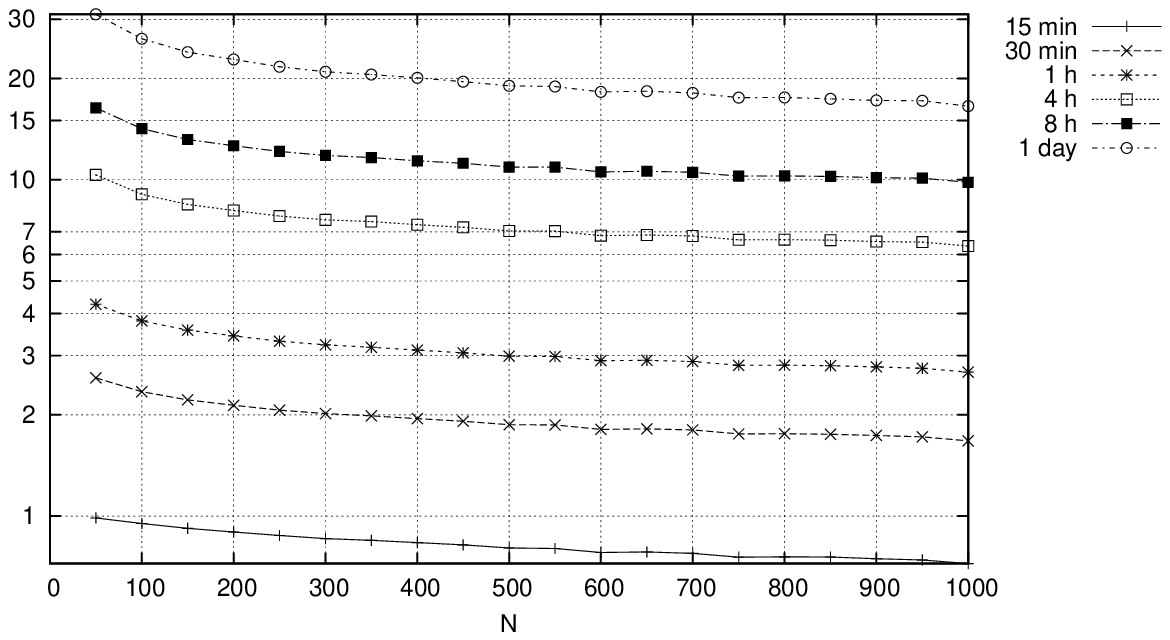}
}
\subfigure[Active appliances]
{
\label{fig:privacy_active_all}
\includegraphics[scale=0.5]{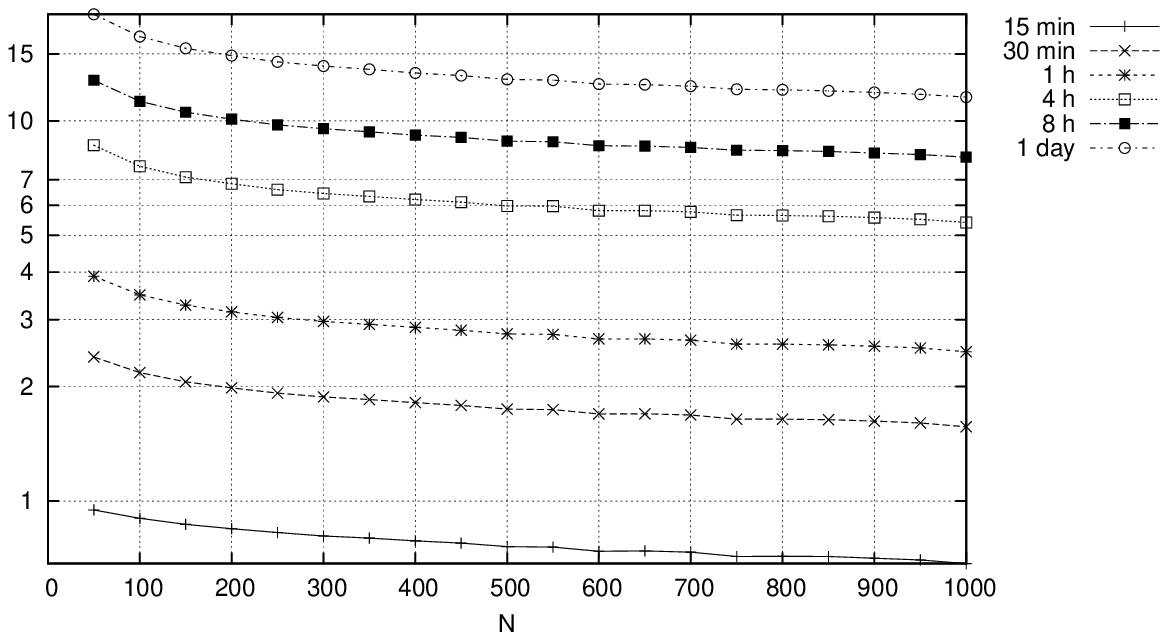}
}
\caption{Privacy of appliances in $s$ long time windows (where $s$ is 10 min, 15 min, 30 min, 1 h, 4 h, 8 h, 1 day).}
\end{figure*}

\subsubsection{Privacy of appliances}
In the previous section, we analysed how a user's privacy  varies over time. In this section, we consider the privacy of the different appliances. For example, we aim at answering the following question: {\em what was the user's privacy when he was watching TV last evening between 18:00 and 20:00?} More specifically, we consider two privacy threats: 
\begin{itemize}
\item \emph{Presence of appliances}: Can the adversary tell that the user watched TV yesterday?
In order to compute the corresponding privacy (i.e. $\varepsilon_s$), we compute $\sum_{t=108}^{120} \varepsilon(t)$, where  $\varepsilon(t) = \{\text{TV's consumption in $t$}\}/\lambda(t)$.
\item \emph{Activation time of appliances}: If the adversary knows that the user watched TV, can he tell what time he did it?
We use statistical inference to detect the position of an appliance signature in the noisy trace. 
\end{itemize}

\paragraph{Presence of appliance:} We summarized some of the appliance privacy in Table \ref{tab:devices} in Appendix \ref{app:table}. Each value is computed by averaging the privacy provided in our 3000 traces.
The appliances can be divided into two major groups: the usage of active appliances indicate that the user is at home and uses the appliance (their consumption significantly changes during their active usage such as iron, vacuum, kettle, etc.), whereas passive appliances (like fridge, freezers, storage heater, etc.) have more or less identical consumption regardless the user is at home or not. 
In general, appliances having lower consumption threats privacy less than devices with higher energy demands.  Obviously, $\varepsilon_s$ increases when $s$ increases since an appliance is used more frequently within longer periods.

Finally, we want to measure the privacy of active appliances. This is equivalent to answer the question if the user was at home in any $s$ long period. The average privacy are depicted in Figure \ref{fig:privacy_active_all}. Observe that there is no considerable differences between Figures \ref{fig:privacy_active_all} and \ref{fig:privacy_all_main}, as a profile is primarily shaped by active appliances (because they typically consume much more than passive appliances). 


\paragraph{Activation time of appliances:}
Consider the consumption profile $\mathbf{V}=(V_1, V_2, \ldots, V_n)$ of a given appliance of a user (on a single day), where the appliance is switched on at $t_{s}$ first and switched off at $t_s+d$ last (i.e., $V_i =0$ for $1 \leq i < t_s$ and $t_s+d < i \leq n$).
The signature of the appliance $\sig(\mathbf{V}) = (V_{t_s}, V_{t_s+1}, \ldots, V_{t_s+d})$ is the consumption profile of the appliance between $t_{s}$ and $t_{s}+d$.
The adversary is provided with the noisy consumption profile of the appliance (i.e., $\hat{\mathbf{V}}$) and, in addition, knows the signature of the appliance, but it does not know $t_{s}$	
(i.e., it knows that the appliance was used with the given signature but does not know when).

The goal of the adversary is to infer the starting slot $t_{s}$ in $\mathbf{V}$ using $\hat{\mathbf{V}}$.
If the adversary's guess is $t'$, the inference accuracy is measured by $|t'-t_{s}|$.
We consider the following adversaries:
\begin{itemize}
\item \rndadv: This is the simple random guesser and serves as a baseline. If there are $n-d$ possible values of $t_{s}$, then the guess $t'$ is selected out of them uniformly at random.
\item \statadv: This adversary knows the relative frequency of each slot occuring as a starting slot (denoted by $f_i$ at slot $i$), and guesses the most likely starting slot: $t' = \max_i f_i$ ($1 \leq i \leq n-d$). This information is publicly available from several surveys \cite{richardson10}.  
\item \badv: This adversary performs bayesian inference on $t_s$. 
In particular, let $\mathbf{V}^t$ denote a profile where the signature starts at slot $t$ (i.e., $\mathbf{V}^t$ is obtained by shifting $\mathbf{V}$ with $|t-t_s|$ positions to left/right if $t-t_s$ is negative/positive.\footnote{More formally, $\mathbf{V}^t=0$ for all $1 \leq i < t$ and $t+d < i \leq n$, and $V_i = \sig(\mathbf{V})_i$ for $1 \leq i \leq t+d$}). 
Assuming that 
the adversary has no prior knowledge about the distribution of starting slots (i.e., they are distributed uniformly at random), the posterior distribution is
computed as 
$$
P(\mathcal{T} = i) =  \frac{\prod_{k=1}^{n}  P(V^i_k + \mathcal{L}(\lambda_k) = \hat{V}_k)}{\sum_{j=1}^{n-d} \prod_{k=1}^{n}P(V^j_k + \mathcal{L}(\lambda_k) = \hat{V}_k)}
$$
where $\mathcal{T}$ describes the posterior distribution of starting slots.
As the bayes risk is "linear" in our case (i.e., $|t'-t_s|$), the bayes' estimate (i.e., $t'$) is the
posterior median (i.e., $t'$ satisfies $P(\mathcal{T} \leq t') \geq 0.5$ and $P(\mathcal{T} \geq t') \leq 0.5$).

\item \bstatadv: We expect better results if the bayesian adversary uses the relative frequencies as a prior
knowledge. In particular, the adversary knows the probability distribution of starting slots a priori, denoted by $\theta = \{f_1, f_2, \ldots, f_{n-d}\}$, which is described by the relative frequencies:
$$
P(\mathcal{T} = i|\theta) = \frac{ \prod_{k=1}^{n} f_i \cdot P(V^i_k + \mathcal{L}(\lambda_k) = \hat{V}_k)}{\sum_{j=1}^{n-d} \prod_{k=1}^{n}P(V^j_k + \mathcal{L}(\lambda_k) = \hat{V}_k)\cdot f_j} 
$$
As before, the bayes' estimate is the posterior median.
\end{itemize}

The inference accuracy of each adversary is shown in Table \ref{tab:devices2} in Appendix \ref{app:table}. The inference is performed on our dataset within a single day.
\bstatadv outperforms all adversaries especially for active devices, however, its accuracy never falls below  1.7 hour. Regarding the passive appliances, \statadv overcomes \bstatadv in general. This is explained by the fact that passive appliances usually follow a regular operation cycle with less user intervention in all households, and the accuracy of \statadv 's is always within the length of one operation cycle independently of the added noise\footnote{The same type of appliance is used in all households.}.

\section{Conclusion}
Our measurements show two different, and conflicting, results. Figure \ref{fig:privacy_all_main} shows that it may actually be 
difficult to hide the presence of activities in a household. 
In fact,  computed $\varepsilon$ values are quite high, even for large clusters. However, results presented
in Tables \ref{tab:devices} and \ref{tab:devices2}  are more encouraging. They show that, although, it might be difficult to hide a user's presence,
it is still possible to hide his actual activity. In fact, appliances privacy bounds ($\varepsilon$ values) are quite small, which indicates
that an adversary will have difficulty telling whether the user is, for example, using his computer or watching TV during a given period
of time. Furthermore, in Table \ref{tab:devices2}, results show that it is even more difficult for an adversary to tell when a given activity actually started. 
Finally, we recall that in order to keep the error $\lambda(t)/\sum_{i=1}^N X_{t}^i$ low while ensuring better privacy one can always 
increase the number of users inside each cluster. For instance, doubling $N$ from 100 to 200 allows to double the
noise while keeping approximately the same error value (0.118 in Figure \ref{fig:cons_cluster} if $\alpha=0$). 
This results in much better privacy, since, on average, doubling the noise halves the privacy parameter $\varepsilon_s$.    

Although more work and research is needed, we believe this is a encouraging  result for privacy. 

\section*{Acknowledgements}
The work presented in this paper was supported in part by the European 
Commission within the STREP WSAN4CIP project. The views and conclusions 
contained herein are those of the authors and should not be interpreted 
as representing the official policies or endorsement of the WSAN4CIP project 
or the European Commission.

\bibliographystyle{abbrv}
\bibliography{paper}

\newpage
\appendix
\section{Proof of Theorem \ref{thm:utility} (Utility)}
\label{app:proof}

\begin{lemma}[Integral property of the Bessel function \cite{kotz01book}]
\label{lem:bessel}
Let 
$$
K_{\vartheta}(x)  = \frac{1}{2} \left( \frac{x}{2} \right)^\vartheta \int_{0}^{\infty} t^{-\vartheta-1} \exp\left(
-t - \frac{x^2}{4t}\right)dt,\qquad x > 0  
$$
define the modified Bessel function of the third kind with index $\vartheta \in \mathbb{R}$.
For any $\gamma>0$ and $\gamma, \nu$ such that $\gamma + 1 \pm \nu > 0$
$$
\int_{0}^{\infty} x^{\gamma} K_{\nu}(ax)dx = \frac{2^{\gamma-1}}{a^{\gamma+1}} \Gamma\left(\frac{1+\gamma+\nu}{2}\right) \Gamma\left(\frac{1+\gamma-\nu}{2}\right)
$$
\end{lemma}

\begin{lemma}
\label{lem:abs}
Let $\mathcal{G}_{1}, \mathcal{G}_{2}$ be i.i.d gamma random variables with parameters $(n, \lambda)$. Then,
\begin{align}
\label{eq:exp}
\mathbb{E}|\mathcal{G}_{1}(n, \lambda) - \mathcal{G}_{2}(n, \lambda)| = \frac{2\lambda}{B\left(\frac{1}{2},\frac{1}{n}\right)} 
\end{align}
and
\begin{align}
\label{eq:var}
\mathit{Var}|\mathcal{G}_{1}(n, \lambda) - \mathcal{G}_{2}(n, \lambda)| = \left( \frac{2}{n} - \frac{4}{B\left(\frac{1}{2}, \frac{1}{n}\right)^2} \right) \lambda^2
\end{align}
where $B(x,y)$ is the beta function defined as $B(x,y) = \frac{\Gamma(x)\Gamma(y)}{\Gamma(x+y)}$.
\end{lemma}

\begin{proof}[of Lemma \ref{lem:abs}]
Consider $\mathcal{Y} = \mathcal{G}_{1} - \mathcal{G}_{2}$. The characteristic function of $\mathcal{Y}$ is
$$
\phi_{\mathcal{Y}}(t) = \left(\frac{1}{1+i\lambda t}\right)^{\frac{1}{n}} \cdot \left(\frac{1}{1-i\lambda t}\right)^{\frac{1}{n}} = \left(\frac{1}{1+\lambda^2 t^2}\right)^{\frac{1}{n}}
$$
which is a special case of the characteristic function of the Generalized Asymetric Laplace distribution (GAL) with parameters $(\theta, \kappa, \omega, \tau)$:
$$
\phi_{\mathit{GAL}}(t) = e^{i\theta t}  \left(\frac{1}{1+i\frac{\sqrt{2}}{2} \omega \kappa t}\right)^{\tau} \cdot \left(\frac{1}{1-i\frac{\sqrt{2}}{2 \kappa} \omega t}\right)^{\tau}
$$
where $\theta = 0, \kappa = 1, \omega = \sqrt{2}\lambda$, and $\tau = 1/n$.
The density function of $GAL(\theta, \kappa, \omega, \tau)$ when $\theta = 0$ and $\kappa = 1$ is
$$
f_{\mathit{GAL}}(x) = \frac{\sqrt{2}}{ \omega^{\tau+1/2} \Gamma(\tau)\sqrt{\pi}} \left( \frac{|x|}{\sqrt{2}}  \right)^{\tau-1/2} K_{\tau-1/2}(\sqrt{2}|x|/\omega)
$$
where $K_{\tau-1/2}(\frac{\sqrt{2}}{\omega}|x|)$ is the Bessel function defined in Lemma \ref{lem:bessel}.
In addition,
$$
\mathbb{E}|\mathcal{Y}| = \int_{-\infty}^{\infty} |x| f_{\mathit{GAL}}(x) dx= 2 \cdot \int_{0}^{\infty} x   \frac{\sqrt{2}}{ \omega^{\tau+1/2} \Gamma(\tau)\sqrt{\pi}} \left( \frac{x}{\sqrt{2}}  \right)^{\tau-1/2} K_{\tau-1/2}(\sqrt{2}x/\omega) dx
$$
which follows from the symmetry property of $f_{\mathit{GAL}}(x)$ ($\phi_{\mathcal{Y}}(t)$ is is real valued). After reformulation, we have
$$
\mathbb{E}|\mathcal{Y}| =  \frac{2\sqrt{2}}{ \sqrt{2}^{\tau-1/2} \omega^{\tau+1/2} \Gamma(\tau)\sqrt{\pi}} \int_{0}^{\infty} x^{\tau+1/2} K_{\tau-1/2}(\sqrt{2}x/\omega) dx
$$
Now, we can apply Lemma \ref{lem:bessel} for the integral and we obtain 
$$
\mathbb{E}|\mathcal{Y}| =  \sqrt{2}\cdot w \cdot \frac{\Gamma(\tau+\frac{1}{2})}{\Gamma(\frac{1}{2})\sqrt{\pi}}  
$$
after simple derivation.
Using that $\sqrt{\pi} = \Gamma(1/2)$ and $B(x,y) = \frac{\Gamma(x)(y)}{\Gamma(x+y)}$, we have
$$
\mathbb{E}|\mathcal{Y}| =  \frac{\sqrt{2}}{B(1/2, \tau)}\cdot w
$$
Applying $\omega = \sqrt{2}\lambda$ and $\tau = 1/n$, we arrive at
Equation \eqref{eq:exp}.

To prove Equation \eqref{eq:var}, consider that
$$
\mathit{Var}(|\mathcal{Y}|) = \mathbb{E}|\mathcal{Y}|^2 - [\mathbb{E}|\mathcal{Y}|]^2
$$
where 
$$ 
\mathbb{E}|\mathcal{Y}|^2 = \mathbb{E}(\mathcal{Y}^2) = \mathbb{E}(\mathcal{G}_{1}^2) + \mathbb{E}(\mathcal{G}_{2}^2) - 2 \cdot \mathbb{E}(\mathcal{G}_{1}) \cdot \mathbb{E}(\mathcal{G}_{2})
$$
Using that $\mathbb{E}(\mathcal{G}_
{1}^2) = \mathbb{E}(\mathcal{G}_{2}^2) = (1/n^2 + 1/n)\lambda^2$,   
we obtain Equation \eqref{eq:var}.

\end{proof}
Now, we can easily prove Theorem \ref{thm:utility}.

\begin{proof}[of Theorem \ref{thm:utility}]
$$
\mathbb{E} | \sum_{i=1}^{N} (X_{t}^{i} - \hat{X}_{t}^{i}) | =
$$
$$
= \mathbb{E}|\sum_{i=1}^{N}\mathcal{G}_{1}(N-M, \lambda) - \sum_{i=1}^{N} \mathcal{G}_{2}(N-M, \lambda)| =
$$
(using that $\sum_{i=1}^{n}\mathcal{G}(k_{i}, \lambda) = \mathcal{G}(1/\sum_{i=1}^{n} \frac{1}{ k_{i}}, \lambda)$)
$$
= \mathbb{E}|\mathcal{G}_{1}(1-M/N, \lambda) - \mathcal{G}_{2}(1-M/N, \lambda)| = 
$$
(using $\alpha = M/N$ and applying Lemma \ref{lem:abs})
$$
= \frac{2}{B(1/2,\frac{1}{1-\alpha})}\lambda
$$
The standard deviation 
$\sqrt{\mathit{Var} | \sum_{i=1}^{N} (X_{t}^{i} - \hat{X}_{t}^{i}) |}$ can be derived identically.

\end{proof}

\section{Privacy of some ordinary appliances}
\label{app:table}
\begin{landscape}
\begin{table*}
\footnotesize

\begin{tabular}{|c|l|l|l|l|l|l|l|l|l|l|l|l|l|l|l|l|}
\hline
& \multirow{2}{*}{\emph{Appliance}}  &\multicolumn{3}{c|}{$s = 30\, \text{min}$} & \multicolumn{3}{c|}{$s = 1\, \text{h}$} & \multicolumn{3}{c|}{$s = 4\, \text{h}$} & \multicolumn{3}{c|}{$s = 8\, \text{h}$} & \multicolumn{3}{c|}{$s = 24\, \text{h}$}   \\
\cline{3-17}
& & \emph{mean} & \emph{dev} & \emph{max} & \emph{mean} & \emph{dev} & \emph{max} & \emph{mean} & \emph{dev} & \emph{max} & \emph{mean} & \emph{dev} & \emph{max} & \emph{mean} & \emph{dev} & \emph{max} \\
\hline
\hline
\multirow{23}{*}
{
\begin{sideways}
\textbf{Active appliances}
\end{sideways}
}

& Lighting & 0.91 & 1.28 & 17.87 & 1.29 & 1.37 & 18.84 & 2.68 & 1.82 & 19.38 & 3.63 & 2.29 & 21.49 & 4.89 & 2.97 & 25.37\\ 
 \cline{2-17}
&  Cassette / CD Player & 0.02 & 0.04 & 0.79 & 0.04 & 0.04 & 0.81 & 0.05 & 0.05 & 0.82 & 0.07 & 0.05 & 0.88 & 0.09 & 0.07 & 0.96\\ 
\cline{2-17}
&  Hi-Fi & 0.10 & 0.17 & 4.43 & 0.16 & 0.19 & 4.59 & 0.17 & 0.20 & 4.62 & 0.18 & 0.21 & 4.62 & 0.19 & 0.21 & 4.62\\ 
\cline{2-17}
&  Iron  &0.75 & 1.81 & 42.91 & 0.82 & 1.82 & 42.99 & 0.92 & 1.83 & 42.99 & 1.00 & 1.86 & 42.99 & 1.02 & 1.89 & 42.99\\ 
\cline{2-17}
&  Vacuum &1.67 & 7.59 & 134.54 & 1.70 & 7.59 & 134.54 & 1.82 & 7.58 & 134.54 & 1.90 & 7.60 & 134.54 & 1.94 & 7.63 & 134.54\\ 
\cline{2-17}
&  Fax &0.04 & 0.10 & 1.55 & 0.04 & 0.10 & 1.55 & 0.04 & 0.10 & 1.55 & 0.05 & 0.10 & 1.56 & 0.05 & 0.10 & 1.56\\ 
\cline{2-17}
&  Personal computer &0.21 & 0.32 & 7.48 & 0.34 & 0.36 & 7.48 & 0.83 & 0.49 & 7.48 & 1.09 & 0.58 & 7.53 & 1.42 & 0.83 & 8.37\\ 
\cline{2-17}
&  Printer &0.07 & 0.30 & 7.78 & 0.08 & 0.31 & 7.78 & 0.09 & 0.31 & 7.78 & 0.10 & 0.31 & 7.78 & 0.11 & 0.31 & 7.83\\ 
\cline{2-17}
&  TV& 0.15 & 0.47 & 7.41 & 0.22 & 0.48 & 7.45 & 0.37 & 0.52 & 7.45 & 0.45 & 0.58 & 8.37 & 0.50 & 0.63 & 8.37\\ 
\cline{2-17}
&  VCR / DVD & 0.05 & 0.16 & 2.81 & 0.07 & 0.17 & 2.84 & 0.10 & 0.17 & 2.89 & 0.13 & 0.18 & 2.95 & 0.14 & 0.19 & 3.01\\ 
\cline{2-17}
&  TV Receiver box &0.03 & 0.11 & 2.12 & 0.05 & 0.11 & 2.21 & 0.08 & 0.12 & 2.32 & 0.10 & 0.13 & 2.40 & 0.11 & 0.14 & 2.42\\ 
\cline{2-17}
&  Hob &1.90 & 9.58 & 132.86 & 1.96 & 9.58 & 132.86 & 2.15 & 9.57 & 132.86 & 2.28 & 9.59 & 132.86 & 2.34 & 9.67 & 132.86\\ 
\cline{2-17}
&  Oven &1.50 & 3.91 & 96.19 & 1.58 & 3.92 & 96.19 & 1.74 & 3.94 & 96.19 & 1.85 & 3.97 & 96.19 & 1.91 & 4.07 & 98.51\\ 
\cline{2-17}
&  Microwave &1.13 & 4.23 & 82.73 & 1.20 & 4.24 & 82.73 & 1.26 & 4.24 & 82.73 & 1.29 & 4.27 & 83.17 & 1.31 & 4.29 & 83.57\\ 
\cline{2-17}
&  Kettle &0.55 & 2.71 & 63.59 & 0.59 & 2.71 & 63.59 & 0.72 & 2.73 & 63.87 & 0.83 & 2.76 & 64.22 & 1.02 & 2.79 & 64.22\\ 
\cline{2-17}
&  Small cooking (group)& 0.25 & 1.61 & 26.16 & 0.25 & 1.61 & 26.16 & 0.26 & 1.61 & 26.16 & 0.27 & 1.61 & 26.16 & 0.27 & 1.62 & 26.16\\ 
\cline{2-17}
&  Dish washer & 0.93 & 2.67 & 55.64 & 1.49 & 2.67 & 55.64 & 1.78 & 2.71 & 55.64 & 1.97 & 2.95 & 60.15 & 2.03 & 2.97 & 60.15\\ 
\cline{2-17}
&  Tumble dryer &2.57 & 8.05 & 152.33 & 3.93 & 8.16 & 154.99 & 5.24 & 8.20 & 155.08 & 6.30 & 8.33 & 155.08 & 7.01 & 8.68 & 155.08\\ 
\cline{2-17}
&  Washing machine &1.23 & 1.43 & 31.57 & 1.30 & 1.45 & 31.72 & 1.96 & 1.63 & 33.24 & 2.55 & 1.76 & 33.24 & 3.07 & 2.07 & 34.62\\ 
\cline{2-17}
&  Washer dryer &1.82 & 1.08 & 19.22 & 3.17 & 1.33 & 19.27 & 4.70 & 1.99 & 25.82 & 6.39 & 2.38 & 25.82 & 7.92 & 3.49 & 33.66\\ 
\cline{2-17}
&  E-INST &1.47 & 1.12 & 6.54 & 1.93 & 1.15 & 6.54 & 3.47 & 1.16 & 7.58 & 4.70 & 1.49 & 9.00 & 7.06 & 2.13 & 10.99\\ 
\cline{2-17}
&  Electric shower &2.13 & 14.78 & 249.24 & 2.16 & 14.78 & 249.24 & 2.28 & 14.78 & 249.24 & 2.34 & 14.78 & 249.24 & 2.38 & 14.80 & 249.24\\ 
\hline

\multirow{6}{*}
{
\begin{sideways}
\textbf{Passive app.}
\end{sideways}
}
&  DESWH & 3.34 & 14.01 & 249.29 & 4.04 & 14.04 & 251.01 & 6.13 & 14.06 & 253.21 & 7.83 & 14.23 & 255.20 & 10.85 & 14.57 & 257.76\\ 

\cline{2-17}
&  Storage heaters &3.22 & 0.32 & 3.96 & 5.64 & 0.56 & 6.95 & 20.20 & 1.99 & 24.87 & 30.45 & 4.23 & 41.48 & 30.45 & 4.23 & 41.48\\ 
  \cline{2-17}

&  Elec. space heating &1.64 & 0.85 & 6.14 & 2.86 & 1.07 & 7.54 & 7.49 & 2.15 & 13.03 & 8.50 & 2.49 & 14.57 & 10.06 & 4.08 & 26.25\\ 
  \cline{2-17}

&   Chest freezer &0.61 & 0.74 & 15.94 & 0.61 & 0.74 & 15.95 & 1.39 & 0.92 & 17.20 & 1.85 & 1.07 & 18.10 & 2.55 & 1.24 & 18.96\\ 
  \cline{2-17}

&  Fridge freezer & 0.91 & 0.39 & 7.56 & 0.91 & 0.40 & 7.61 & 2.19 & 0.95 & 8.67 & 2.94 & 1.25 & 10.58 & 4.07 & 1.61 & 11.69\\ 
  \cline{2-17}

&  Refrigerator &0.44 & 0.22 & 3.83 & 0.45 & 0.23 & 4.00 & 1.06 & 0.49 & 4.77 & 1.40 & 0.64 & 5.68 & 1.92 & 0.80 & 6.50\\ 
  \cline{2-17}

&  Upright freezer &0.67 & 0.39 & 8.37 & 0.67 & 0.39 & 8.42 & 1.63 & 0.80 & 9.09 & 2.16 & 1.03 & 10.99 & 2.98 & 1.31 & 11.98\\ 
\hline

\end{tabular}
\caption{$\varepsilon_{s}$ of different appliances in case of different $s$. $N=100$ and the sampling period is 10 min.}
\label{tab:devices}
\end{table*}
\end{landscape}

\begin{table*}
\footnotesize
\begin{tabular}{|c|l|l|l|l|l|l|l|l|l|l|l|l|}
\hline
& \multirow{2}{*}{\emph{Appliance}}  & \multirow{2}{*}{\# of users} & \multicolumn{2}{c|}{\rndadv} & \multicolumn{2}{c|}{\statadv} & \multicolumn{2}{c|}{\badv} & \multicolumn{2}{c|}{\bstatadv}\\ 
\cline{4-11}
& & & \emph{mean} & \emph{dev} & \emph{mean} & \emph{dev} &  \emph{mean} & \emph{dev} & \emph{mean} & \emph{dev} \\
\hline
\hline
\multirow{23}{*}
{
\begin{sideways}
\textbf{Active appliances}
\end{sideways}
}

& Lighting  & 2998 & 2.53 & 2.54 & 5.16 & 3.66 & 1.87 & 1.73 & \textbf{1.71} & \textbf{2.34}\\  
 \cline{2-11}
&  Cassette / CD Player  & 2650 & 4.70 & 4.09 & 3.34 & 3.70 & 3.96 & 2.50 & \textbf{3.19} & \textbf{3.46}\\ 
\cline{2-11}
&  Hi-Fi & 744 & 7.49 & 5.49 & 4.10 & 3.29 & 5.58 & 3.29 & \textbf{4.04} & \textbf{2.65}\\ 
\cline{2-11}
&  Iron  & 1247 & 6.53 & 4.40 & 3.89 & 3.19 & 3.62 & 2.94 & \textbf{2.95} & \textbf{2.28}\\ 
\cline{2-11}
&  Vacuum & 1192 & 6.61 & 4.47 & 4.00 & 3.22 & 3.54 & 3.02 & \textbf{2.92 }& \textbf{2.45}\\ 
\cline{2-11}
&  Fax & 241 & 6.85 & 5.66 & 7.78 & 4.81 & 5.76 & 3.26 & \textbf{4.19} & \textbf{2.85}\\
\cline{2-11}
&  Personal computer & 1970 & 5.35 & 4.50 & 5.32 & 4.55 & 4.79 & 3.20 & \textbf{4.03} & \textbf{3.49}\\ 
\cline{2-11}
&  Printer & 1608 & 6.21 & 5.06 & 5.73 & 4.69 & 4.71 & 3.01 & \textbf{4.07} & \textbf{3.04}\\ 
\cline{2-11}
&  TV & 2519 & 5.41 & 4.07 & 4.22 & 3.41 & 3.73 & 2.44 & \textbf{2.50} & \textbf{2.50}\\ 
\cline{2-11}
&  VCR / DVD & 2299 & 5.55 & 4.09 & 4.29 & 3.44 & 3.72 & 2.42 & \textbf{2.53} & \textbf{2.57}\\ 

\cline{2-11}
&  TV Receiver box & 2413 & 5.58 & 4.09 & 4.27 & 3.42 & 3.71 & 2.37 & \textbf{2.53} & \textbf{2.58}\\ 
\cline{2-11}
&  Hob & 857 & 6.53 & 4.49 & 3.64 & 3.19 & 3.55 & 2.89 & \textbf{2.95} &\textbf{2.48}\\ 
\cline{2-11}
&  Oven & 760 & 6.31 & 4.50 & 3.78 & 3.13 & 3.35 & 2.99 & \textbf{2.74} & \textbf{2.41}\\
\cline{2-11}
&  Microwave  & 505 & 6.41 & 4.24 & 3.96 & 3.17 & 3.39 & 2.97 & \textbf{2.90} & \textbf{2.44}\\ 
\cline{2-11}
&  Kettle  & 2808 & 4.81 & 4.13 & 3.62 & 3.84 & 3.83 & 2.67 & \textbf{3.29} & \textbf{3.48}\\ 
\cline{2-11}
&  Small cooking (group) & 1441 & 6.55 & 4.41 & 3.92 & 3.18 & 3.51 & 2.65 & \textbf{3.00} & \textbf{2.40}\\ 

\cline{2-11}
&  Dish washer  & 434 & 6.32 & 4.46 & 4.57 & 3.39 & 3.28 & 3.00 & \textbf{2.71} & \textbf{2.19}\\ 

\cline{2-11}
&  Tumble dryer  & 1018 & 5.79 & 4.15 & 4.32 & 3.37 & 2.23 & 2.56 & \textbf{2.03} & \textbf{2.57}\\ 

\cline{2-11}
&  Washing machine  & 2228 & 5.28 & 4.02 & 3.58 & 3.31 & 2.85 & 2.67 & \textbf{2.36} & \textbf{2.80}\\ 

\cline{2-11}
&  Washer dryer  & 417 & 5.05 & 3.77 & 3.26 & 3.07 & 1.94 & 2.21 & \textbf{1.79} & \textbf{2.66}\\ 

\cline{2-11}
&  E-INST  & 29 & 3.05 & 2.42 & 1.87 & 3.09 & 1.84 & 2.35 & \textbf{1.71} & \textbf{2.88}\\ 
 
\cline{2-11}
&  Electric shower & 1039 & 6.10 & 4.31 & 3.76 & 3.12 & 3.47 & 2.96 & \textbf{2.89} & \textbf{2.36}\\ 
\hline

\multirow{6}{*}
{
\begin{sideways}
\textbf{Passive app.}
\end{sideways}
}
&  DESWH  & 510 & 3.70 & 3.41 & 2.54 & 3.14 & \textbf{1.22} & \textbf{1.73} & 1.54 & 2.46\\ 

\cline{2-11}
&  Storage heaters  & 84 & 8.50 & 5.84 & \textbf{0.00} & \textbf{0.00} & 0.27 & 0.25 & \textbf{0.00} & \textbf{0.00}\\ 

  \cline{2-11}

&  Elec. space heating  & 73 & 6.52 & 5.05 & 6.75 & 5.02 & 2.85 & 3.42 & \textbf{2.14} & \textbf{3.14}\\ 
  \cline{2-11}

&   Chest freezer  & 466 & 0.56 & 0.47 & 0.51 & 0.44 & 0.42 & 0.32 & \textbf{0.40} & \textbf{0.34}\\
  \cline{2-11}

&  Fridge freezer  & 1954 & 0.49 & 0.42 & \textbf{0.28} & \textbf{0.30} & 0.34 & 0.29 & 0.34 & 0.33\\ 

  \cline{2-11}

&  Refrigerator  & 1301 & 0.56 & 0.48 & \textbf{0.35}& \textbf{0.39} & 0.40 & 0.33 & 0.41 & 0.42\\ 
  \cline{2-11}

&  Upright freezer  & 866 & 0.55 & 0.46 & \textbf{0.35} & \textbf{0.37} & 0.38 & 0.30 & 0.37 & 0.36\\ 
 \hline
\hline

\end{tabular}
\caption{Inference accuracy of starting slots. $N=100$, $T_p=10$ min, and "\# of users" means the number of users who 
have the given appliance in our dataset. The accuracy ($|t'-t_s|$) is given in hours.}
\label{tab:devices2}
\end{table*}

\end{document}